\documentclass[12pt,a4paper,notitlepage]{article}
\usepackage[utf8]{inputenc}
\usepackage[english]{babel}
\usepackage[T1]{fontenc}
\usepackage{hyperref}
\usepackage{amsmath}
\usepackage{amsthm}
 
\usepackage{amsfonts}
\usepackage{amssymb}
\usepackage{color} 
\usepackage{caption}
\usepackage{array,multirow,makecell}
\setcellgapes{1pt}
\makegapedcells
\newcolumntype{R}[1]{>{\raggedleft\arraybackslash }b{#1}}
\newcolumntype{L}[1]{>{\raggedright\arraybackslash }b{#1}}
\newcolumntype{C}[1]{>{\centering\arraybackslash }b{#1}}
\newcommand{\Tr}{\mathrm{Tr}}

\newtheorem{definition}{Definition}
\newtheorem{proposition}{Proposition}
\newtheorem{remark}{Remark}

\newcommand{\sym}{\mathrm{Sym}}

\usepackage{enumerate}
\usepackage{subfig}

\setcellgapes{1pt}
\makegapedcells
\newcolumntype{R}[1]{>{\raggedleft\arraybackslash }b{#1}}
\newcolumntype{L}[1]{>{\raggedright\arraybackslash }b{#1}}
\newcolumntype{C}[1]{>{\centering\arraybackslash }b{#1}}

\newcommand{\cF}{{\mathcal F}}
\newcommand{\cG}{{\mathcal G}}

\newcommand{\cL}{{\mathcal L}}

\newcommand{\beq}{\begin{equation}}
\newcommand{\eeq}{\end{equation}}
\newcommand{\bea}{\begin{eqnarray}}
\newcommand{\eea}{\end{eqnarray}}
\definecolor{mygray}{gray}{0.3}

\newcommand{\bes}{\begin{eqnarray}}
\newcommand{\ees}{\end{eqnarray}}

\newcommand\restr[2]{{
  \left.\kern-\nulldelimiterspace 
  #1 
  \vphantom{\big|} 
  \right|_{#2} 
  }}

\usepackage{multicol}
\usepackage{amssymb}
\usepackage{graphicx}

\makeatletter
\begin{document}
\begin{flushleft}
Progress in Group Field Theory   
and \\
Related Quantum Gravity Formalism \\
\end{flushleft}
\begin{center}
\textbf{\Large{
Progress in  solving  nonperturbative  renormalization 
group  for tensorial group field theory }}
\vspace{15pt}

{\large Vincent Lahoche$^a$\footnote{vincent.lahoche@th.u-psud.fr},   and Dine Ousmane Samary$^{a,b}$\footnote{dine.ousmanesamary@cipma.uac.bj}}
\vspace{0.5cm}

a)\,Commissariat à l'\'Energie Atomique (CEA, LIST),
 8 Avenue de la Vauve, 91120 Palaiseau, France

c)\,  Facult\'e des Sciences et Techniques/ ICMPA-UNESCO Chair, Universit\'e d'Abomey-
Calavi, 072 BP 50, Benin


\end{center}

\vspace{5pt}
\begin{abstract}
\noindent
This manuscript  aims at giving new advances on the functional renormalization group applied to the tensorial group field theory. It is based on a series of our three papers [arXiv:1803.09902],  [arXiv:1809.00247] and [arXiv:1809.06081]. We consider the polynomial Abelian $U(1)^d$ models without closure constraint. More specifically, we discuss  the case of the quartic melonic  interaction. We present a new approach, namely the effective vertex expansion method, to solve the exact Wetterich flow equation, and investigate the resulting flow equations, specially regarding the existence of non-Gaussian fixed points for their connection with phase transitions. To complete this method, we consider a non-trivial constraint arising from the Ward-Takahashi identities, and discuss the disappearance of the global non-trivial fixed points taking into account this constraint. Finally, we argue in favor of an alternative scenario involving a first order phase transition into the reduced phase space given by the Ward constraint. 
\end{abstract}

\setcounter{tocdepth}{2}
\section{Introduction}
In seeking a theory to unify modern physics,  i.e. a well defined theory of quantum gravity, numerous contributions have been made. Despite the fact that none of them has given a complete resolution to the problem, several major advances have been observed. In the number of these advances, we count the very recent propositions such as loop quantum gravity \cite{Rovelli:1997yv}-\cite{Rovelli:1998gg}, dynamical triangulation \cite{Ambjorn:1992eh}-\cite{Ambjorn:2013tki}, noncommutative geometry \cite{Connes:1990qp}-\cite{Aastrup:2006ib}, group field theories (GFTs) \cite{Oriti:2006ar}-\cite{Oriti:2014yla}  and  tensors models (TMs) \cite{Gurau:2009tw}-\cite{Gurau:2013pca}. These approaches  are considered as new background independent approaches according to several theoreticians.
GFTs are quantum field theories over the group manifolds and  are considered as the second quantization version of loop quantum gravity \cite{Oriti:2014yla}.  These theories are caracterized by  the specific form of non-locality in their interactions. TMs, especially colored ones, allow one to define probability measures on simplicial pseudo-manifolds such that the tensor of rank $d$ represents a $(d-1)$-simplex.  TMs admit the large $N$-limit ($N$ is the size of the tensor) dominated by the graphs called melons, thanks to the Gurau breakthrough \cite{Gurau:2011xq}-\cite{Gurau:2013pca}. The large $N$-limit or the  leading order encodes a sum over a class of colored triangulations of the $D$-sphere and its  behaviour is a powerful tool which allows us to understand the continuous limit of these models through, for instance, the study of critical exponents and phase transitions. TM and GFT are combined to give birth to a new class of field theories called tensorial group field theory. These class of field models enjoy renormalization and asymptotic freedom \cite{Carrozza:2012uv}-\cite{Carrozza:2014rba}. Using the functional renormalization group (FRG) method, it is also possible  to identify  the equivalent of Wilson-Fisher fixed point for some particular cases of models. 

There are several ways to introduce the FRG in field theories. The first approach is the one pioneered by Wilson, simple and intuitive and therefore yields a powerful way to think about quantum field theories \cite{Wilson}. This method allows a smooth interpolation between the known microscopic laws IR-regime and the complicated macroscopic phenomena in physical systems UV-regime and is constructed with the incomplete integration as cutoff procedure. Well after Polchinski provided a new approach called Wilson-Polchinski FRG equation  \cite{Polchinshi} to address the same question inspired from the Wilson's method. This very practicable method, which may be integrated with an arbitrary cutoff function and expanded up to the next to leading order of the derivative expansion. Despite the fact that all these approaches seem to be nonperturbative, in practice, the perturbative solution has appeared more attractive.  More recently the so called Wetterich flow equation \cite{Wetterich:1992yh}, is proposed to study the nonperturbative FRG and this study requires  approximations or truncations and numerical analysis which is not very well controlled. The FRG equation allows to determine the fixed points and probably the phase transition. These phase transitions in the case of TGFT models may help to identify the emergence of general relativity and quantum mechanics through the pregeometrogenesis scenario \cite{Oriti:2018dsg}-\cite{Wilkinson:2015fja}. Indeed, the way the quantum
degrees of freedom are organized to shape a geometric structure which can be identified with a semi-classical space-time is one of the challenges for GFT approach. In the geometrogenesis point of view, the standard space-time geometry is understood as an emergent property, the scenario leading to this geometric limit being assumed quite closed to Bose Einstein condensation in condensed matter physics.

In the recent works \cite{Lahoche:2018ggd}-\cite{Lahoche:2018vun} the effective vertex expansion method is used in the context of the FRG. This leads to the definition of new class of   equations called structure equations that help to solve the Wetterich flow equations. Taking into account the leading order contribution in the symmetric phase,  the non-perturbative regime without truncation can be studied. The Ward-Takahashi (WT) identities is also derived and become a constraint along the flow. Note that the WT-identities are universal for all field theories having a symmetry, and are not specific to TGFT. Therefore all the fixed points must belong inside to the domain of this constraint line, before being considered as an acceptable fixed points. In the case of quartic melonic TGFT models it has been shown that the fixed point occurring from the solution of Wetterich equation violates this constraint for any choice of the regulator function. This violation is also independent of the method used to find this fixed point, whether it is the truncation, or the EVE method. This point will be discussed carefully in this note. Let us remark that most of the TGFT models  previously studied in literature are showed to admit at least a non trivial fixed point and therefore a phase transition. 
The phase transitions are very useful to the likely emergence of metric and are linked to the existence of fixed points, which becomes unavoidable  in the search of models which may be probably describe our universe after geometrogenesis scenario.  However, in this paper, we study  the quartic  $T^4$-TGFT  models and prove that no fixed points can be found. First of all we considered the Wilson-Polchinski renormalization group method and show the weakness of this method  in the nonperturbative regime.  Then we consider  the nonperturbative Wetterich flow equation from which the nonperturbative analysis can be made by an approximation on the average effective action called truncation. The EVE method is used to get around the approximation and therefore solves the flow without truncation. The set of Ward-Takahashi identities and structure equations are derived to provide a nontrivial constraint on the reliability of the approximation schemes, i.e. the truncation and the choice of the regulator.

The paper is organized as follows: In section \eqref{sec2} we recall the FGR method by Wilson-Polchinski and apply it in the context of TGFT.  Despite the efficiencies of this method, we will present some questions that arise, in the search of a  nonperturbative solution and then we will go further in the Wetterich flow equation. Section \eqref{sec3} is dedicated to the description of the Wetterich flow equation and the corresponding solution when the truncation method is applied. We also show that the only nontrivial fixed point which comes from the solution of the flows,
violates the Ward identities.  In the section \eqref{sec4}, we perform new nonperturbative analysis using the so called structure equations is given and the solution of the flow equations are also derived.  In the last section \eqref{sec5} we provide a discussion and conclusion to our work.

\section{Introduction to the nonperturbative renormalization  for TGFT}\label{sec2}
FRG is a powerful ingredient to think about when it comes to quantum field theories. Generally, in every situation where the scale belong to a range of correlated variables, the theory  may be treated by the RG. The first conceptual framework is Wilson's version of the RG which, by Polchinski, may be applied in the case of quantum field theory.
In this section we discuss the nonperturbative renormalization group using not only the Wilson-Polchinski equation but also the Wetterich flow equation. We discuss each method and  consider the Wetterich flow equation as more suitable for the treatment of FRG applied to TGFT.
Thanks to the Wilson method, the renormalization and renormalization group are understood as a \textit{coarse-graining} process from a microscopic theory toward an effective long-distance theory. There are in fact different implementations of this idea, depending on the context. 
In the context of TGFT, we consider the pair of complex fields $\phi$ and $\bar\phi$  which takes values of $d$-copies of arbitrary group $G$:
\bea
\phi,\bar \phi: G^d\rightarrow \mathbb{C}.
\eea
In a particular case we assume that $G=U(1)$ is an Abelian compact Lie group. For the rest we only consider the Fourier transform of the fields $\phi$ and $\bar\phi$ denoted by $T_{\vec p}$ and $\bar{T}_{\vec p}$ respectively,  $\vec p\in \mathbb{Z}^d$ written as (for $\vec g\in U(1)^{d}$, $g_{j}=e^{i\theta_{j}}$):
\bea
\phi(\vec \theta\,)=\sum_{\vec p\in\mathbb{Z}^{d}}\,T_{\vec p} \,e^{i\sum_{j=1}^d\theta_{j}p_{j}},\quad
\bar\phi(\vec \theta\,)=\sum_{\vec p\in\mathbb{Z}^{d}}\,\bar T_{\vec p}\, e^{-i\sum_{j=1}^d\theta_{j}p_{j}}.
\eea
The description of the statistical field theory is given by the partition function $\mathcal Z[J,\bar J]$:
\bea\label{path}
\mathcal Z[J,\bar J]=\int d\mu_{C}\, e^{-S_{int}+\langle J,\bar T\rangle+\langle T,\bar J \rangle},
\eea
where $S_{int}$ is the interaction functional action assumed to be tensor invariant, $J$, $\bar J$ the external currents and $\langle J,\bar T\rangle$ a shorthand notation for
\begin{equation}
\langle J,\bar T\rangle:= \sum_{\vec{p}} J_{\vec{p}} \bar{T}_{\vec{p}}\,.
\end{equation}
The Gaussian measure $d\mu_{C}$ is then fixed with the choice of the covariance $C$. In this paper, we adopt a Laplacian-type propagator of the form: 
\bea\label{propa1}
C(\vec p\,)=\frac{1}{\vec p\,^2+m^2}=\int\,d\mu_C \,T_{\vec p}\,\bar T_{\vec p} \,.
\eea
In order to prevent the UV divergences and supress the high momenta contributions, the propagator \eqref{propa1} has to be regularized. In usual  case the Schwinger regularization is used:
\bea\label{propa2}
C_\Lambda(\vec p\,)=\frac{e^{-(\vec p\,^2+m^2)/\Lambda^2}}{\vec p\,^2+m^2}.
\eea
In general case, by defining  the function $\vartheta(t)$  such  that the condition $|1-\vartheta(t)|\leq Ce^{-\kappa t}$ is fatisfied for $C,\kappa>0$ and $t\rightarrow +\infty$, we can write the propagator as a Laplace transform:
\begin{equation}\label{regularization}
C_{\Lambda}(\vec{p})= \int_{0}^{+\infty}\, dt\, \vartheta(t\Lambda^2)\,e^{-t(\vec{p}^{\,2}+m^2)}.
\end{equation}
Then, we shall make the simplest choice $\vartheta(t)=\Theta(t-1)$, where $\Theta(t)$ is the Heaviside function, in order to recover the Schwinger regularization \eqref{propa2}. For the rest we keep mind that the propagator is regularized and the infinite limit will be given in an appropriate way. In this case the following result in well satisfied:
\begin{proposition}\label{prop1}
Let us consider two non-normalized Gaussian measures $d\mu_{C}$ and $d\mu_{C'}$ 
whose covariances $ C$ and $C' $ are related by  $C'=C+\Delta$ and such that $C$, $C'$ and $\Delta $ are  assumed to be  positive. Then we get the following relation:
\begin{align}
\int d\mu_{C}(\bar{T}_1,T_1)d\mu_{\Delta}(\bar{T}_2,T_2)e^{-S_{int}(T_1+T_2, \bar{T}_1+\bar{T}_2)}=\left(\dfrac{\det(\Delta C)}{\det(C')}\right)^{1/2}\int d\mu_{C'}(\bar{T},T)e^{-S_{int}(T, \bar{T})},
\end{align}
where $T=T_1+T_2$ and $\bar T=\bar T_1+\bar T_2$.
\end{proposition}
\begin{proof}
The proof of this formula can simply be given using the definition of the Gaussian measure $d\mu_C$ with mean zero and covariance matrix $C$ as
\bea
d\mu_C=\det(\pi C)^{-\frac{1}{2}}e^{-\langle T,C^{-1}\bar T\rangle} dT\,d\bar T.
\eea
and the fact that
\bea
\int\, d\mu_{C'}(T,\bar T) \,e^{-\langle J,\bar T\rangle-\langle T,\bar J\rangle}&=&e^{\langle J,C' \bar J\rangle}=e^{\langle J,C \bar J\rangle}e^{\langle J,\Delta \bar J\rangle}.
\eea
\end{proof}

We introduce tensorial unitary invariants, or simply tensorial  invariants.  An invariant is a polynomial  $P(T,\bar T)$ in the tensor entries $T_{\vec p}$ and $\bar T_{\vec p}$ which is invariant under the following action of $U(N)^{\otimes d}$ ($N$ being the size of the tensors): 
\bea
T_{\vec p}\rightarrow \sum_{\vec q}\, U^{(1)}_{p_1q_1}\cdots U^{(d)}_{p_dq_d} T_{\vec q},\quad \bar T_{\vec p}\rightarrow \sum_{\vec q}\, \bar U^{(1)}_{p_1q_1}\cdots \bar U^{(d)}_{p_dq_d} \bar T_{\vec q}
\eea
The algebra of invariant polynomials is generated by a set of polynomials labelled as bubbles.
A bubble is a connected, bipartite graph, regular of degree
$d$, whose edges must be colored with a color belonging to the set 
$\{1,\cdots, d\}$,  and such that all $d$ colors are incident at each vertex (and is incident to exactly once). Examples of bubbles are displayed in Fig.~\ref{fig:Bubbles}. 
\begin{figure}\centering
\begin{minipage}[b]{45mm}
\includegraphics[scale=.7]{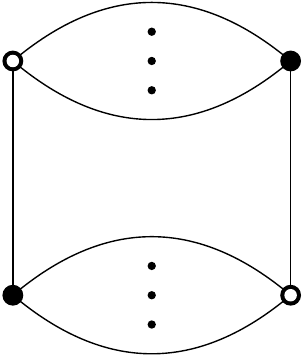}
\end{minipage}
\caption{The 4-vertex bubble   from  which  the dots indicate multiple edges.}\label{fig:Bubbles}
\end{figure}

In this paper we consider the quartic melonic  $T_5^4$ model which is proved to be renormalizable  in all orders  in the perturbative theory.  The interaction   of this  model  takin into account the leading order contributions: (melon) is  written graphically as:
\bea\label{s4}
S_{int}=
\lambda_{41}\sum_{i=1}^5\,\vcenter{\hbox{\includegraphics[scale=0.8]{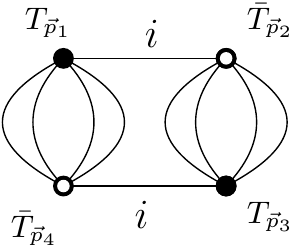} }}
\eea
Note that the interaction \eqref{s4} is  invariant under the unitary transformations ${\bf U}\in U^{\otimes d}$. In contrast, it is not the case for the kinetic terms and sources terms due to the non-trivial propagator and sources $J$ and $\bar J$. This implies the existence of a non-trivial Ward-identity which becomes a strong constraint and will be taking into account in the FRG point of view.

\subsection{Wilson-Polchinski equation}
In this subsection we discuss the Wilson-Polchinski RG equation and provide the corresponding solutions of the quartic melonic TGFT.
 For this let us  introduce a dilatation parameter $s<1$. This parameter will be used as an evolution parameter in the integration around the UV modes. The RG idea is that if we want to describe the phenomena at scales down to $s$, then we should be able to use the set of variables defined at the scale $s$.  Indeed, define the variation
\begin{align}
\Delta_{s,\Lambda}(\vec{p})&:=C_{\Lambda}(\vec{p})-C_{s\Lambda}(\vec{p})\\\nonumber
&\,=\int_{0}^{+\infty} dt \int_{s^2}^{1}dx\dfrac{d}{dx}\vartheta(tx\Lambda^2)e^{-t(\vec{p}^{\,2}+m^2)}.
\end{align}
In the case where $s$ is closed to $1$, denoting by $D_{s,\Lambda}(\vec{p})$ the infinitesimal version of the above variation, we get:
\begin{align}
\Delta_{s,\Lambda}(\vec{p})&\simeq \dfrac{2(1-s)}{\Lambda^2}e^{-(\vec{p}^{\,2}+m^2)/\Lambda^2}=:(1-s)D_{s,\Lambda}(\vec{p})\,,
\end{align}
such that the partition function can be written as an integral over two fields, respectively associated to the ``slow'' and ``rapid'' modes. Starting with the partition function $\mathcal{Z}_{\Lambda}$ at scale $\Lambda$, we get
\begin{equation}
\mathcal{Z}_{\Lambda}[S_{int}]:=\int d\mu_{C_{\Lambda}}(\bar{T},T)e^{-S_{int,\Lambda}(T, \bar{T})}.
\end{equation}
The proposition \eqref{prop1} allows us to decompose $\mathcal{Z}_{\Lambda}[S_{int}]$ into two Gaussian integrals over two fields, $T_>$ and $T_<$, corresponding respectively to the ``rapid'' and ``slow'' modes,  with covariances $\Delta_{s,\Lambda}$ and $C_{s\Lambda}$:
\begin{align}
\mathcal{Z}_{\Lambda}[S_{int}]=&\left(\dfrac{\det(\Delta_{s,\Lambda} C_{s\Lambda})}{\det(C_{\Lambda})}\right)^{-1/2}\int d\mu_{C_{s\Lambda}}(\bar{T}_<,T_<)\label{decomp}\int d\mu_{\Delta_{s,\Lambda}}(\bar{T}_>,T_>)e^{-S_{int}(T_<+\bar{T}_>, \bar{T}_<+\bar{T}_>)}\,.
\end{align}
Then, identifying the effective action $S_{int,s\Lambda}$ at scale $s\Lambda$ as:
\begin{align}\label{effectiveaction}
&e^{-S_{int,s\Lambda}(T_<,\bar{T}_<)}:=\frac{1}{\sqrt{\det \Delta_{s,\Lambda}}}\int d\mu_{\Delta_{s,\Lambda}}(\bar{T}_>,T_>)e^{-S_{int}(T_<+T_>, \bar{T}_<+\bar{T}_>)}\,,
\end{align}
and the decomposition \ref{decomp} becomes:
\begin{equation}\label{effectiveaction2}
\mathcal{Z}_{\Lambda}=\left(\dfrac{\det C_{s\Lambda}}{\det C_{\Lambda}}\right)^{-1/2}\int d\mu_{C_{s\Lambda}}(\bar{T}_<,T_<)e^{-S_{int,s\Lambda}(T_<,\bar{T}_<)}\,.
\end{equation}
Now, for an infinitesimal step, keeping only the leading order terms in $1-s$ when $s$ is very close to $1$, we find:
\begin{align}
&\qquad e^{-\Delta S_{int,\Lambda}(T_<,\bar{T}_<)}\label{bubu}=1-\Tr\Big[\Big(\frac{\delta^2 S_{int,\Lambda}}{\delta T \delta \bar{T}}-\frac{\delta S_{int,\Lambda}}{\delta T} \frac{\delta S_{int,\Lambda}}{\delta \bar{T}}\Big)\Delta_{s,\Lambda}\Big]+\mathcal{O}(1-s),
\end{align}
with $\Delta S_{int,\Lambda}(T_<,\bar{T}_<):=S_{int,s\Lambda}(T_<,\bar{T}_<)-S_{int\Lambda}(T_<,\bar{T}_<)$. At the same time, expanding the left hand side of \ref{effectiveaction2} in powers of $1-s$, and identifying the power of $1-s$ leads to :
\begin{equation}
\dfrac{dS_{int,s\Lambda}}{ds}=-\Tr\bigg\{\Big(\frac{\delta^2 S_{int,s\Lambda}}{\delta T \delta \bar{T}}-\frac{\delta S_{int,s\Lambda}}{\delta T} \frac{\delta S_{int,s\Lambda}}{\delta \bar{T}}\Big)D_{s,\Lambda}\bigg\}\label{eqflow}.
\end{equation}
Graphically this equation is given by (and is considered as the Wilson-Polchinski RG equation):
\bea\label{eqflow2new}
\frac{d}{ds} \vcenter{\hbox{\includegraphics[scale=0.9]{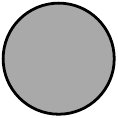} }}=\Tr\Bigg[\vcenter{\hbox{\includegraphics[scale=0.9]{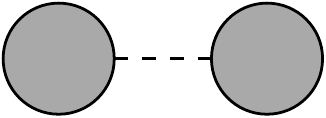} }}-\vcenter{\hbox{\includegraphics[scale=0.9]{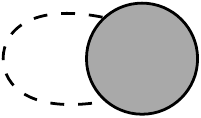} }}\Bigg].
\eea
Note that we may consider $\Lambda$ not only  as a fundamental scale, but also as an arbitrary step on the flow, meaning that the equation \ref{eqflow} holds at each step of the flow. Physically,  equation \ref{eqflow} explains how the couplings are affected when the fundamental scale changes, and is therefore the one pioneered  idea  of the renormalization group flow  firstly given by Wilson.  This approach follows from a remarkably simple and intuitive idea and yields a very powerful way to think about quantum field theories. The relation \eqref{eqflow2new}  can be also expanded in the following result:
\begin{proposition}
The set of Wilson-Polchinski renormalization group equations are given by
\bea\label{WPnew}
\dfrac{d\mathcal{V}^{(n_l)}}{ds}=-\sum_{\vec{p}\vec{\bar{p}}}D_{s,\Lambda,\vec{p}\vec{\bar{p}}}\frac{\partial}{\partial \bar{T}_{\vec{p}}}\frac{\partial}{\partial T_{\vec{\bar{p}}}}\mathcal{V}^{(n_l+1)}+\sum_{n_m=0}^{n_l-1}\sum_{\vec{p}\vec{\bar{p}}}D_{s,\Lambda,\vec{p}\vec{\bar{p}}}\frac{\partial\mathcal{V}^{(n_m+1)}}{\partial \bar{T}_{\vec{p}}}\frac{\partial\mathcal{V}^{(n_{l}-n_m)}}{\partial T_{\vec{\bar{p}}}}-n_l\eta_s\mathcal{V}^{(n_l)}\,,
\eea
where $D_{s,\Lambda,\vec{p}\vec{\bar{p}}}=D_{s,\Lambda}(\vec{p})\delta_{\vec{p}\vec{\bar{p}}}$, $\eta_s:=\dfrac{d}{ds}\ln Z(s)$. In this formula we denote by $n_l$ the number of black and white nodes in each interactions and  we consider the following expansion for $S_{int,s\Lambda}[T,\bar{T}]$: 
\begin{align}
S_{int,s\Lambda}[T,\bar{T}]=\sum_{n_l}\mathcal{V}^{(n_l)}=\sum_{n_l}\sum_{\{\vec{p}_{i},\vec{\bar{p}}_{i}\}}\mathcal{V}^{(n_l)\,\vec{\bar{p}}_{1},...,\vec{\bar{p}}_{l}}_{\vec{p}_{1},...,\vec{p}_{l}}\prod_{i=1}^{l}T_{\vec{p}_{i}}\bar{T}_{\vec{\bar{p}}_{i}}.
\end{align}

\end{proposition}
\begin{proof}
A pragmatic way to introduce field strength renormalization is the following. We consider a wave function  $Z(s)$ and the regularized field $T=Z(s)^{\frac{1}{2}}\tilde{T}$ at the scale $s\Lambda$. A new functional $\tilde{S}_{int,s\Lambda}$ is associated to this field such as $\tilde{S}_{int,s\Lambda}[\tilde{T},\bar{\tilde{T}}]=S_{int,s\Lambda}[{T},\bar{{T}}]$. The equation \ref{eqflow} is then modified into (we deleted  the tildes notation):
\begin{align}
\dfrac{dS_{int,s\Lambda}}{ds}=&-\Tr\bigg\{\Big(\frac{\delta^2 S_{int,s\Lambda}}{\delta T \delta \bar{T}}-\frac{\delta S_{int,s\Lambda}}{\delta T} \frac{\delta S_{int,s\Lambda}}{\delta \bar{T}}\Big)D_{s,\Lambda}\bigg\}\nonumber\\
&-\frac{1}{2}\eta_s\bigg[\Tr\Big(\frac{\delta S_{int,s\Lambda}}{\delta T}T\Big)+\Tr \Big(\bar{T}\frac{\delta S_{int,s\Lambda}}{\delta \bar{T}}\Big)\bigg].\label{eqflow2} 
\end{align}
Then, by  considering  the following expansion for $S_{int,s\Lambda}[T,\bar{T}]$: 
\begin{align}
S_{int,s\Lambda}[T,\bar{T}]=\sum_{n_l}\mathcal{V}^{(n_l)}=\sum_{n_l}\sum_{\{\vec{p}_{i},\vec{\bar{p}}_{i}\}}\mathcal{V}^{(n_l)\,\vec{\bar{p}}_{1},...,\vec{\bar{p}}_{l}}_{\vec{p}_{1},...,\vec{p}_{l}}\prod_{i=1}^{l}T_{\vec{p}_{i}}\bar{T}_{\vec{\bar{p}}_{i}},
\end{align} 
we get the relation \eqref{WPnew}.
 \end{proof}
The Wilson-Polchinski  equation is a leading order equation in the perturbation rather than the loop expansion.
Note that we  can show that this equation can be turned into a Fokker-Planck equation and therefore may be formally solved by a standard method. The rest of this section is devoted to a perturbative analysis of the flow equations.
Before starting this computation, we have to precise the approximation regime. We shall consider only the UV limit which corresponds to the higher values of the scale parameter $s$ or  to the higher momenta variables $\vec p$   or  also for the smaller distances, and we assume that $s\Lambda$ and $\Lambda$ are large. However, the analysis in the UV regime can be extended to IR limit, which corresponds to the smaller values of the scale parameter $s$. 
More precisely, our approximation can be characterized by both $s\Lambda$ and $\Lambda$  in the UV and by $s\Lambda/\Lambda$ in the IR. 
At scale $\Lambda$, and up to contributions of order $\lambda_{41}^2$, kipping only the melonic contribution  the action providing from \eqref{s4} is assumed to be of the form
\begin{align}
S^4_{int,s\Lambda}[\bar{T},T]=\delta m^2 \sum_{\vec{p}}\bar{T}_{\vec{p}}T_{\vec{p}}+\delta Z\sum_{\vec{p}}\vec{p}^{\,2}\bar{T}_{\vec{p}}T_{\vec{p}}+\lambda_{41} \sum_{i=1}^5 \sum_{\{\vec{p}_i,\vec{q}_i\}}\mathcal{W}^{(i)}_{\vec{p}_1,\vec{q}_1;\vec{p}_2,\vec{q}_2}T_{\vec{p}_1}T_{\vec{p}_2}\bar{T}_{\vec{q}_1}\bar{T}_{\vec{q}_2}, \label{actionin}
\end{align}
where the first two terms take into account the fact that the parameter of Gaussian measure, the mass and the Laplacian term, can be affected by the integration of the UV modes, and these counter-terms, assumed to be of order $\lambda_{41}$, take into account these modifications. The vertex  $\mathcal{W}^{(i)}_{\vec{p}_1,\vec{q}_1;\vec{p}_2,\vec{q}_2}$ is  a product of delta function and is given by
\bea
\mathcal{W}^{(i)}_{\vec{p}_1,\vec{q}_1;\vec{p}_2,\vec{q}_2}=\delta_{p_{1i}q_{2i}}\delta_{q_{1i}p_{2i}}\prod_{j\neq i}\delta_{p_{1j}q_{1j}}\delta_{p_{2j}q_{2j}}.
\eea
Moreover, note that in this approach the corrections to the Laplacian term are not suppressed by an effective counter-term in the action, but absorbed in the wave function renormalization. It is fixed such that all the Laplacian corrections are canceled by the $\eta_s$ term in the RG equation for $\mathcal{V}^{(1)}$.
We adopt the standard Ansatz, namely that the generic interaction of valence $n$ are of order $\lambda^{n/2-1}_{41}$. This allows to organize systematically the perturbative solution, for which we shall construct the $\lambda^2_{41}$ order. 

\subsubsection*{$\mathcal{V}^{(1)}$ at order $\lambda_{41}$}
The first corrections occur at order $\lambda_{41}$ for $\mathcal{V}^{(1)}$, whose flow equation write as:
\begin{align}\label{flowdeg1}
\Big(\dfrac{d}{ds}+\eta_s\Big)\mathcal{V}^{(1)}=-4\lambda_{41}\sum\limits_{\substack{\vec{p}_1,\vec{q}_1\\\vec{p}_2,\vec{q}_2}}D_{s\,\Lambda\,\vec{p}_1,\vec{\bar{p}}_1}\sym{\mathcal{W}^{(i)}_{\vec{p}_1,\vec{q}_1;\vec{p}_2,\vec{q}_2}}T_{\vec{p}_2}\bar{T}_{\vec{q}_2},
\end{align}
where 
\bea
\sym{\mathcal{W}_{\vec{p}_1,\vec{q}_1;\vec{p}_2,\vec{q}_2}}=\mathcal{W}_{\vec{p}_1,\vec{q}_1;\vec{p}_2,\vec{q}_2}+\mathcal{W}_{\vec{p}_2,\vec{q}_1;\vec{p}_1,\vec{q}_2}
\eea
and $\sym\mathcal{W}:=\sum_i\sym\mathcal{W}^{(i)}$ and $\mathcal{W}=\sum_{i=1}^6\mathcal{W}^{(i)}$. The r.h.s involves two typical contributions which are pictured graphically in figure \ref{fig8c4}, where the contraction with $D_{s,\Lambda}$ is represented by a dotted line with a gray box. 
\begin{center}
\includegraphics[scale=1]{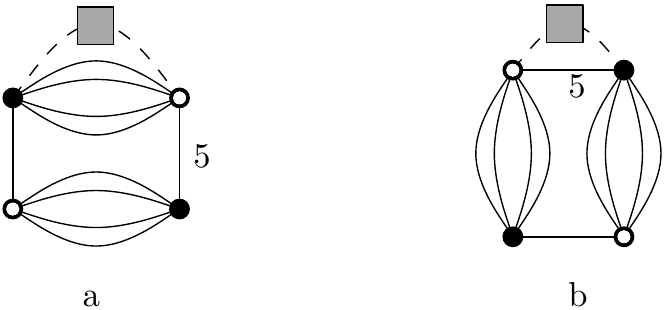} 
\captionof{figure}{Typical graphs contributing to the interaction $\mathcal{V}^{(1)}$ of degree 2.}
\label{fig8c4}
\end{center}
In the UV limit that we consider, the non-melonic contractions of type \ref{fig8c4}b, creating only one internal face (of color 5 in this figure), can be neglected in comparison to the melonic contributions of the form of figure \ref{fig8c4}a. Retaining only the melonic contractions, equation \ref{flowdeg1} becomes:
\begin{align}\label{exacttwopoints}
\Big(\dfrac{d}{ds}+\eta_s\Big)\mathcal{V}^{(1)}=-2\lambda_{41}\sum\limits_{\substack{\vec{p}_1,\vec{q}_1\\\vec{p}_2,\vec{q}_2}}D_{s\,\Lambda,\,\vec{p}_1\vec{q}_1}{\mathcal{W}}_{\vec{p}_1,\vec{q}_1;\vec{p}_2,\vec{q}_2}T_{\vec{p}_2}\bar{T}_{\vec{q}_2},
\end{align}
with $D_{s,\Lambda}=dC_{s\Lambda}/ds$.
Expanding this relation  in power of $p_{5}$, we generate mass and wave function corrections, and also the sub-dominant corrections, involving powers of $p_{5}$ greater than two. They correspond to the first deviation to the original form  \ref{actionin}. Neglecting these sub-dominant contributions, we get the expansion
\begin{align}\label{sum1}
\sum_{p_1,...,p_4}\dfrac{2}{s^3\Lambda^2}&e^{-\frac{1}{(s\Lambda)^2}(\vec{p}^{\,2}+m^2)}\sim 2\pi^2 s\Lambda^2-\frac{2\pi^2}{s}(p_5^2+m^2)+\mathcal{O}(s),
\end{align}
for which we only keep the leading order terms in $s$, we can extract the dominant contributions to the mass and wave-function renormalization. The term in $p_5^2$ generates a non-local 2-point interaction of the form $-\delta Z(s) \Tr(\bar{T}\Delta_{\vec{g}}T)$, where $\Delta_g$ is the Laplacian on $U(1)^{\times 5}$, and the first term generates a mass correction. Summing over the five colors, we find, at first order in $\lambda_{41}$:
\begin{equation}
\eta_s=\dfrac{4\pi^2 \lambda_{41}}{s},\quad
\dfrac{d }{ds}\delta m^2=-4\pi^2 \lambda_{41}  s\Lambda^2+\dfrac{4\pi^2 \lambda_{41}}{s}m^2.
\end{equation}

\subsubsection*{$\mathcal{V}^{(3)}$ and $\mathcal{V}^{(2)}$ at order $\lambda^2_{41}$}
Let us focus on the second order perturbative solution i.e. at  $\lambda^2_{41}$ in which we have to take into account the contributions of interactions of valence six, $\mathcal{V}^{(3)}$, verifying the flow equation:
\begin{equation}
\dfrac{d\mathcal{V}^{(3)\,\vec{q}_{1},\vec{q}_{2},\vec{q}_{3}}_{\vec{p}_{1},\vec{p}_{2},\vec{p}_{3}}}{ds}=4\lambda^2_{41}\sum_{i,j,\vec{p},\vec{q}}\mathcal{W}^{(i)}_{\vec{p}_1,\vec{q}_1,\vec{p},\vec{q}_2}\mathcal{W}^{(j)}_{\vec{p}_2,\vec{q}_3;\vec{p}_3,\vec{q}}\,D_{s,\Lambda,\vec{p}\vec{q}},
\end{equation}
which can be easly integrated with  the initial condition  $\mathcal{V}^{(3)\,\vec{q}_{1},\vec{q}_{2},\vec{q}_{3}}_{\vec{p}_{1},\vec{p}_{2},\vec{p}_{3}}(1)=0$ as:
\begin{align}\label{v3}
\mathcal{V}^{(3)\,\vec{q}_{1},\vec{q}_{2},\vec{q}_{3}}_{\vec{p}_{1},\vec{p}_{2},\vec{p}_{3}}(s)=&-4\lambda^2_{41}\sum_{i,j,\vec{p},\vec{q}}\mathcal{W}^{(i)}_{\vec{p}_1,\vec{q}_1,\vec{p},\vec{q}_2}\mathcal{W}^{(j)}_{\vec{p}_2,\vec{q}_3;\vec{p}_3,\vec{q}} \,\big(C_{\Lambda}-C_{s\Lambda}\big)_{\vec{p}\,\vec{q}}\,.
\end{align}
As for the interaction of degree $1$, the structure of this effective interaction can be understood as a contraction between two bubbles, as pictured in figure \ref{fig6c4}, where the dotted line with a gray box represents the contraction with $C_{\Lambda}-C_{s\Lambda}$. 
\begin{center}
\includegraphics[scale=1.4]{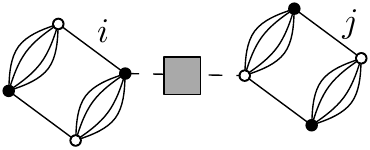} 
\captionof{figure}{Typical graph contributing to the interaction of $\mathcal{V}^{(3)}$ of degree 6.}.
\label{fig6c4}
\end{center}

Let us now  build the effective coupling for the quartic melonic interaction at  order $\lambda^2_{41}$, for which we shall extract only  the leading behavior. From the Wilson-Polchinski  flow equations \eqref{WPnew}, it seems that the coupling evolution receives many contributions in which the  first one comes from $\mathcal{V}^{(3)}$. Now deriving two times this interaction with respect to the fields, we obtain an interaction of degree two, which can be either 1PI, when the contraction with $D_{s\,\Lambda}$ links two black and white nodes of two different bubbles, or one particle reducible (1PR) if the two nodes stand on the same interaction bubble. Explicitly we get
\begin{align}
\nonumber\Big[\frac{d}{ds}+2\eta_s-&4\delta m^2\bar{D}_{s,\,\Lambda}[\{p_i\},\{q_i\}]\Big]\lambda_{41}\mathcal{W}^{(i)}_{\vec{p}_2,\vec{q}_1;\vec{p}_3,\vec{q}_2}=4\lambda^2_{41}\sum_{\vec{p},\vec{q}\vec{p}^{\,\prime},\vec{q}^{\,\prime}}
\Bigg[\bar{\sym}\Big(\mathcal{W}^{(i)}_{\vec{p}^{\,\prime},\vec{q}_1;\vec{p},\vec{q}_2}\mathcal{W}^{(i)}_{\vec{p}_2,\vec{q}^{\,\prime};\vec{p}_3,\vec{q}}\Big)\\
&+2\sum_j\bar{\sym}\Big(\mathcal{W}^{(i)}_{\vec{p}_2,\vec{q}_1;\vec{p},\vec{q}_2}\mathcal{W}^{(j)}_{\vec{p}^{\,\prime},\vec{q}^{\,\prime};\vec{p}_3,\vec{q}}\Big)\Bigg]\times\big(C_{\Lambda}-C_{s\Lambda}\big)_{\vec{p}\,\vec{q}}\,D_{s,\Lambda,\vec{p}^{\,\prime}\vec{q}^{\,\prime}}\,, \label{exactebetafunct}
\end{align}
where:
\begin{align}
&\bar{\sym}\Big(\mathcal{W}^{(i)}_{\vec{p}^{\,\prime},\vec{q}_1;\vec{p},\vec{q}_2}\mathcal{W}^{(j)}_{\vec{p}_2,\vec{q}^{\,\prime};\vec{p}_3,\vec{q}}\Big):=\mathcal{W}^{(i)}_{\vec{p}^{\,\prime},\vec{q}_1;\vec{p},\vec{q}_2}\mathcal{W}^{(j)}_{\vec{p}_2,\vec{q}^{\,\prime};\vec{p}_3,\vec{q}}+\mathcal{W}^{(i)}_{\vec{p},\vec{q}_1;\vec{p}^{\,\prime},\vec{q}_2}\mathcal{W}^{(j)}_{\vec{p}_3,\vec{q}^{\,\prime};\vec{p}_2,\vec{q}},
\end{align}
and
\begin{equation}
\bar{D}_{s,\,\Lambda}[\{p_i\},\{q_i\}]:=D_{s,\,\Lambda}(\vec{p}_2)+D_{s,\,\Lambda}(\vec{q}_1)+D_{s,\,\Lambda}(\vec{p}_3)+D_{s,\,\Lambda}(\vec{q}_2).
\end{equation}
Equation \ref{exactebetafunct} gives the exact behavior for the beta function at order $\lambda^2_{41}$, but we can easily  see that it reduces to the expression of the beta function already obtained for the one loop computation in the deep UV sector. Indeed, retaining only the melonic contributions, and noting that 1PR contributions of the r.h.s are exactly canceled by the term involving the mass correction $\delta m$ in the l.h.s, we get:
\begin{align}
\Big[\frac{d}{ds}+2&\eta_s\Big]\lambda_{41}\mathcal{W}^{(i)}_{\vec{p}_2,\vec{p}_3;\vec{q}_1,\vec{q}_2}\approx 4\lambda^2_{41}\sum_{\vec{p},\vec{q}\vec{p}^{\,\prime},\vec{q}^{\,\prime}}\mathcal{W}^{(i)}_{\vec{p}^{\,\prime},\vec{p};\vec{q}_1,\vec{q}_2}\times \mathcal{W}^{(i)}_{\vec{p}_2,\vec{p}_3;\vec{q}^{\,\prime},\vec{q}} \big(C_{\Lambda}-C_{s\Lambda}\big)_{\vec{p}\,\vec{q}}D_{s,\Lambda,\vec{p}^{\,\prime}\vec{q}^{\,\prime}}.
\end{align}\label{couplingevolve}
The computation of the loop appearing on the r.h.s leads to
\begin{align}\label{sum2}
\sum_{p_1,...,p_4}\int_{1}^{s} ds' \dfrac{4}{s'^3s^3\Lambda^4}&e^{-\big(\frac{1}{(s\Lambda)^2}+\frac{1}{(s'\Lambda)^2}\big)(\vec{p}^{\,2}+m^2)}\sim -\frac{\pi^2}{s}+\mathcal{O}(s),
\end{align}
from which we finally deduce that:
\begin{equation}
s\dfrac{d \lambda_{41}}{ds}=-4\pi^2\lambda^2_{41}\label{effective1}
\end{equation}
which, as claimed before, is exactly the value of the one-loop beta function already obtained in the one loop computation of the beta function.\\

\noindent
We conclude that the main advantage of the Wilson-Polchinski equation is that it provides a very well defined  interpretation of the renormalization group flow in the space of couplings. However, except for perturbative computations, the Wilson-Polchinski equation is more adapted to mathematical and formal proofs than to non-perturbative analysis.  The analysis  beyond the perturbative level requires another formulation of the coarse-graining renormalization group, called \textit{Wetterich equation}, which allows usually to better capture the non-perturbative effects. The price to pay is an approximation scheme a bit more difficult to use. This non-perturbative approach to the renormalization group flow will be the subject of the next  sections. 

\section{Wetterich flow equation}\label{sec3}

The Wetterich method and its incarnation into the FRG approach is a set of techniques allowing to go beyond the difficulties coming from the Wilson-Polchinski equation, in particular in regard to track non-perturbative aspects. The Wetterich equation is a first-order functional integro-differential equation for the effective action. The central object of the method is a continuous set of models labelled with a real parameter $s$ running from UV scales ($s \to +\infty$) to the IR scales ($s\to -\infty$). The physical running scale $e^s$ define for each models what is UV and what is IR, the fluctuation with a large size with respect to the referent scale (the UV fluctuations) being integrated out. The renormalization group equation then describes how the coupling constant change when the referent scale change. To say more, each model is characterized by a specific partition function $\mathcal{Z}_s$, labeled by  $s$ and defined as:
\begin{equation}
\mathcal{Z}_s[J,\bar{J}]:= \int d\mu_C \, e^{-S_{int}(T,\bar{T})+R_s[T,\bar T] +\langle J,\bar T\rangle+\langle T,\bar J \rangle}\,.
\end{equation}
As a result, the original model corresponds to $R_s[T,\bar T]=0$, and because physically this limit have to match with the IR limit $e^s\to 0$, we require that $R_{s}[T,\bar T]$ vanish in the same limit. The term $R_s[T,\bar T]$ called \textit{IR regulator} play the same role as a momentum dependent mass term, becoming very large in the UV and vanishing in the IR. It is chosen ultra-local in the usual sense:
\begin{equation}
R_{s}[T,\bar T]:= \sum_{\vec{p}} \bar{T}_{\vec{p}} \,r_s(\vec{p}\,)T_{\vec{p}}\,,
\end{equation}
the regulating function $r_s(\vec{p}\,)$ being chosen to satisfy the boundary conditions in the UV/IR limit. Moreover, for $s$ fixed, $r_s$ aims at freezing the long distance fluctuations, which are discarded from the functional integration. In formula: $r_s(\vec{p}\,) \to 0$ for $\vert \vec{p}\,\vert /e^s \to 0$, and $r_s(\vec{p}\,) \gg 1$ in the opposite limit. \\

\noindent
The object whose we track the evolution is called \textit{effective averaged action} $\Gamma_s$, defined as (slightly modified version of) the Legendre transform of the standard free energy $\mathcal W_s=\ln \mathcal{Z}_s$:
\bea
\Gamma_s[M,\bar M]=\langle \bar J, M\rangle+\langle \bar M, J\rangle-\mathcal W_s[J,\bar J]-R_s[M,\bar M] \,.
\eea
This definition ensures that $\Gamma_s$ satisfies the physical boundary conditions
$
\Gamma_{s=\ln\Lambda}=S,\,\, \Gamma_{s=-\infty}=\Gamma,
$
where $\Lambda$ denote some fundamental UV cutoff.
The fields $M$ and $\bar M$ are the mean values of $T$ and $\bar T$ respectively and are given by
\bea
M=\frac{\partial\mathcal W}{\partial \bar J},\quad \bar M=\frac{\partial \mathcal W}{\partial J}
\eea
where $\mathcal{W}:=\mathcal{W}_{s=-\infty}$.
In general the regulator $r_s$ is chosen to be 
$
r_s=Z(s) k^2 f\Big(\frac{\vec p\,^2}{k^2}\Big),\,\, k=e^s,
$ and such that the  boundary conditions  is well satisfied.such that the boundary conditions in the UV/IR limit are well satisfied. Taking the first derivative with respect to the flow parameter $s$, one can deduce the Wetterich equation, describing the behavior of the effective action $\Gamma_s$ when $s$ changes:
\bea\label{Wetterich}
\partial_s\Gamma_s=\Tr\, \partial_s r_s (\Gamma_s^{(2)}+r_s)^{-1}\,,
\eea
where $\Gamma^{(2)}_s$ denotes the second order partial derivative of $\Gamma_s$ with respect to the mean fields $M$ and $\bar M$. This equation is exact, but generally impossible to be solved exactly. A large part of the FRG approach is then devoted to approximate the exact trajectory of the RG flow. In this review, we will discuss two methods, the truncation method, and the effective vertex expansion method. \\

This section is especially devoted to the truncations. The general strategy is to cut crudely in the full theory space, projecting systematically the flow into the interior of a finite dimensional subspace. To say more, the average effective action is chosen to be of the form:
\bea
\Gamma_s=Z(s)\sum_{\vec p\in \mathbb{Z}^d} T_{\vec p}(\vec p\,^2+ m^2(s))\bar{T}_{\vec p}+\sum_{n}^{N} \,\lambda_n V_{n}(T,\bar T)
\eea
where $N$ is \textit{finite}, $V_n$ stands for the interaction function of order $n$ and $ m^2$ and $ \lambda_n$ are the mass and coupling constants. With this truncation and with an appropriate regulator it is possible to solve the Wetterich flow equation \eqref{Wetterich}.
 In the case of quartic melonic interaction and by taking the standard modified  Litim's regulator:
\bea\label{Litim}
r_s(\vec p\,)=Z(s)(e^{2s}-\vec p\,^2)\Theta(e^{2s}-\vec p\,^2)
\eea
the Wetterich equation  can be solved analytically and the phase diagram may be given \cite{Lahoche:2018ggd}-\cite{Lahoche:2018oeo}, \cite{Lahoche:2018vun}. The corresponding non trivial fixed points can be studied taking into account the behavior of the flow around these points. Note that the validity of the fixed point require a few analysis taking into account the Ward-Takahashi identities as a new constraint along the flow line. The full violation of this constraint for quartic melonic interaction make this class of fixed points unphysical .  We discuss this point in detail in this section (for more detail see subsection \eqref{WID}). The flow equations are
\bea\label{flownew}
\left\{\begin{array}{llll}
\dot m^2&=-2d\lambda_{41} I_2(0)\\
\dot Z(s)&=-2\lambda_{41} I_2'(q=0)\\
\dot{\lambda}_{41}&=4\lambda^2_{41} I_3(0)
\end{array}\right.
\eea
with the renormalization condition
\bea\label{rencond}
m^2(s)=\Gamma_s^{(2)}(\vec p=\vec 0),\quad \lambda_{41}(s)=\frac{1}{4}\Gamma_s^{(4)}(\vec 0,\vec 0,\vec 0,\vec 0).
\eea
where
\bea
I_n(q)=\sum_{\vec p\in\mathbb{Z}^{(d-1)}}\frac{\dot r_s}{(Z(s)\vec p\,^2+Zq^2+m^2+r_s)^n}.
\eea
Explicitly using the integral representation of the above sum and with $d=5$, $\eta=\dot{Z}/Z$ we get
\bea
I_n(0)=\frac{\pi^2 e^{6s-2ns}}{6Z(s)^{n-1}(\bar m^2+1)^n}(\eta+6),\quad I'_n(0)=-\frac{\pi^2 e^{4s-2ns}}{2Z(s)^{n-1}(\bar m^2+1)^n}(\eta+4).
\eea
In order to get an autonomous system, the standard strategy consist at extracting from the couplings the part coming from their own scaling, defining their \textit{canonical dimension}. Strictly speaking, fields, couplings and all the parameters involved in the theory are dimensionless, because there are no referent space-time, and then not referent scale. The canonical dimension emerge taking into account quantum corrections, and is usually defined as the optimal scaling, with respect to the UV cut-off of the quantum corrections. Conversely, it can be defined as the scaling transformation allowing to get an autonomous system. Note that these two points of views are note strictly equivalent, especially with respect to the choice of the initial content of the theory. For our purpose however, the two strategy provides exactly the same rescaling, and in term of dimensionless parameter $\lambda_{41}=:Z^2 \bar \lambda_{41}$, $m^2=:e^{2s}Z\bar m^2$  the system \eqref{flownew}  becomes
\begin{align}
\left\{
    \begin{array}{ll}
       \beta_m&=-(2+\eta)\bar{m}^{2}-2 d\bar{\lambda}_{41}\,\frac{\pi^2}{(1+\bar{m}^{2})^2}\,\left(1+\frac{\eta}{6}\right)\,, \\
       \beta_{41}&=-2\eta \bar{\lambda}_{41}+4\bar{\lambda}_{41}^2 \,\frac{\pi^2}{(1+\bar{m}^{2})^3}\,\left(1+\frac{\eta}{6}\right)\,, \label{syst2}
    \end{array}
\right.
\end{align}
where $\beta_m:= \dot{\bar{m}}^{2}$,  $\beta_{41}:=\dot{\bar{\lambda}}_{41}$ and:
\begin{equation}
\eta:=\frac{4\bar{\lambda}_{41} \pi^2}{(1+\bar{m}^{2\alpha})^2-\bar{\lambda}_{41}\pi^2}\,.\label{etatruncated}
\end{equation}
 The solutions of the system \eqref{syst2} is given analytically : 
\bea
p_\pm=\Big(\bar m^2_{\pm}=-\frac{23\mp\sqrt{34}}{33},\bar\lambda_{41,\pm}=\frac{328\mp8\sqrt{34}}{11979\pi^2}\Big).
\eea
Numerically
\bea
p_+=(-0.52,0.0028),\quad p_-=(-0.87,0.0036).
\eea
Apart from the fact that we have a singularity line around  the point $\bar m^2=-1$ in the flow equation \eqref{flownew},  another  second singularity arise from the anomalous dimension denominator, and corresponds to a line of singularity, with equation:
\bea
\Omega(\bar m,\bar\lambda_{41}):=(\bar m^2+1)^2-\pi^2\bar\lambda_{41}=0
\eea
This line of singularity splits the two dimensional phase space of the truncated theory into two connected regions characterized by the sign of the function $\Omega$. The region $I$, connected to the Gaussian fixed point for $\Omega>0$ and the region $II$ for $\Omega<0$. For $\Omega=0$, the flow becomes ill defined. The existence of this singularity is a common feature for expansions around vanishing means field, and the region $I$ may be viewed as the domain of validity of the expansion in the symmetric phase. Note that to ensure the positivity of the effective action, the melonic coupling must be positive as well. Therefore, we expect that the physical region of the reduced phase space correspond to the region $\lambda_{41}\geq 0$. From definition of the connected region $I$ and because of the explicit expression \eqref{etatruncated}, we deduce that :
\begin{equation}
\eta \geq 0\,,\qquad \text{In the symmetric phase}\,.
\end{equation}
Then, only the fixed point $p_+$ is taking into account.  In the next subsection we will discuss the  violation of the Ward identity around this fixed point $p_+$, and then clarify our analysis given in  \cite{Lahoche:2018ggd}.   The phase diagram is given in the figure \eqref{figflow1}
\begin{center}
\includegraphics[scale=0.6]{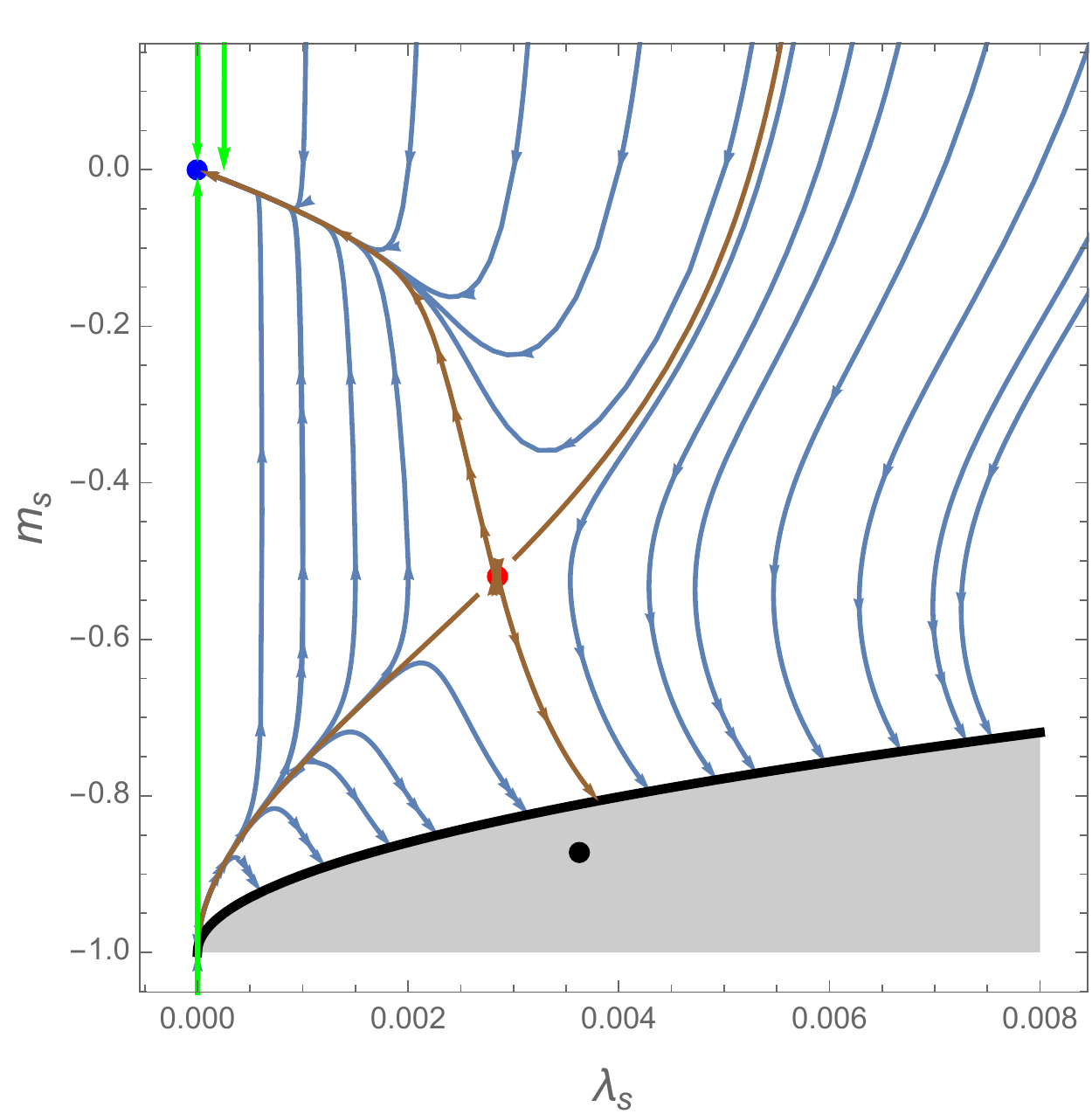} 
\captionof{figure}{Renormalization group flow trajectories around the relevant fixed points obtained from a numerical integration. The Gaussian fixed point and the first non-Gaussian fixed point are respectively in blue and in red, and  the last fixed point is in black. This fixed point is in the grey region bounded by the singularity line corresponding to the denominator of $\eta$. Finally, in green and brown we draw the eigendirections around Gaussian and non-Gaussian fixed points respectively. Note that arrows this fixed point  the flow are oriented from IR to UV.}\label{figflow1}
\end{center}
\subsection{Convenient search of the Ward identities}\label{WID}
Let $\mathcal{U}=(U_1,U_2,\cdots, U_d)$, where the $U_i\in U_\infty$ are infinite size unitary matrices in momentum representation. We define the transformation:
\bea
\mathcal{U}[T]_{\vec p}&=& \sum_{\vec q}U_{1\,,p_1q_1}U_{2\,,p_2q_2}\cdots U_{d\,,p_dq_d}    T_{\vec q\,},\label{transform}
\eea
such that the interaction term is invariant i.e.
$
\mathcal{U}[S_{int}]=S_{int}\,.
$
Then consider an infinitesimal transformation: 
\begin{equation}
\mathcal{U}=\mathbb{\textbf{I}}+\vec{\epsilon},\quad \vec{\epsilon}=\sum_i\mathbb{I}^{\otimes (i-1)}\otimes \epsilon_i\otimes \mathbb{I}^{\otimes(d-i)}\,,
\end{equation}
where $\mathbb{I}$ is the identity on $U_\infty$, $\mathbb{\textbf{I}}=\mathbb{I}^{\otimes d}$ the identity on $U_\infty^{\,\otimes d}$, and $\epsilon_i$ denotes skew-symmetric hermitian matrix such that $\epsilon_i=-\epsilon_i^\dagger$ and 
$
\vec{\epsilon}_i[T]_{\vec{p}}={\epsilon_i}_{p_iq_i}T_{p_1,\cdots,q_i,\cdots, p_d}$.
The  invariance of the path integral \eqref{path} means $\vec{\epsilon}\,[\mathcal Z_s[J,\bar{J}]]=0$, i.e.:
\begin{equation}\label{sat1}
\vec{\epsilon}\,[\mathcal Z_s[J,\bar{J}]]=\int dT d\bar{T} \bigg[ \vec{\epsilon}\,[S_{kin}]+\vec{\epsilon}\,[S_{int}]+\vec{\epsilon}\,[S_{source}]\bigg] e^{-S_s[T,\bar{T}]+\langle \bar{J},T\rangle+\langle \bar{T},J\rangle} =0.\,
\end{equation}
Computing each term separately, we get successively using linearity of the operator $\vec{\epsilon}$:
\bea
&&\vec{\epsilon}\,[S_{int}]=0\,,\label{satnon1}
\\
&&\vec{\epsilon}\,[S_{source}]=-\sum_{i=1}^d\sum_{\vec{p}, \vec{q}\,} \prod_{j\neq i} \delta_{p_jq_j} [\bar{J} _{\vec{p}}\,T_{\vec{q}\,}-\bar{T}_{\vec{p}}{J} _{\vec{q}\,}]{\epsilon_i}_{p_iq_i}\,,\label{satnon2}
\\
 &&\vec{\epsilon}\,[S_{kin}]=\sum_{i=1}^d\sum_{\vec{p}, \vec{q}} \prod_{j\neq i} \delta_{p_jq_j} \bar{T}_{\vec{p}}\big[C_s(\vec{p}\,^{2})-C_s(\vec{q}\,^{2})\big]T_{\vec{q}}\,\,\epsilon_{ip_iq_i}\,,\label{sat2}
\eea
where $\prod_{j\neq i} \delta_{p_jq_j}:=\delta_{\vec p_{\bot_i}\vec q_{\bot_i}}$, $p_{\bot_i}:=\vec p\setminus\{p_i\}$, $C_s^{-1}=C_{-\infty}^{-1}+r_s$ and $C_{-\infty}^{-1}=Z_{-\infty}\vec p\,^2+m_{-\infty}^2$.  $Z_{-\infty}$ is the renormalized  wave function  usually denoted by $Z$.
We get the following result:
\begin{proposition} \label{propnew}\,
The ward identity gives relation between two and four point functions as:
\bea
\sum_{\vec r_{\bot_i},\vec s_{\bot_i}}\delta_{\vec r_{\bot_i}\vec s_{\bot_i}}(C^{-1}_s(\vec r)-C^{-1}_s(\vec s))\langle T_{\vec r}\bar T_{\vec s} T_{\vec p}\bar T_{\vec q}\rangle=-\delta_{\vec p_{\bot_i}\vec q_{\bot_i}}(G_s(p)-G_s(q))\delta_{ r_{i} s_{i}},
\eea
where, defined by $\Gamma^{(4)}_s$, the 1PI four point function, we get
\bea
\langle T_{\vec r}\bar T_{\vec s} T_{\vec p}\bar T_{\vec q}\rangle=\Gamma^{(4)}_{s,\vec r\vec s;\vec p\vec q}\Big(G_s(\vec p)G_s(\vec q)+\delta_{\vec r\vec p}\delta_{\vec s\vec q}\Big)G_s(\vec r)G_s(\vec s)
\eea
\end{proposition}
\begin{proof}
 The formal invariance of the  path integral  implies that the variations of these terms have to be compensate by a non trivial variation of the source terms. 
 Combining the two expressions \eqref{sat1}, \eqref{satnon1}, \eqref{satnon2} and  \eqref{sat2}, we come to 
\begin{align}
\sum_{i=1}^d\sum_{\vec{p}_{\bot_i}, \vec{q}_{\bot_i}} \delta_{\vec p_{\bot_i}\vec q_{\bot_i}} \bigg[\frac{\partial}{\partial J_{\vec{p}} }\big[C_s(\vec{p}\,^{2})-C_s(\vec{q}\,^{\,{2}})\big]\frac{\partial}{\partial \bar{J}_{\vec{q}\,}}-\bar{J} _{\vec{p}}\,\frac{\partial}{\partial \bar{J}_{\vec{q}\,}}+{J} _{\vec{q}\,}\frac{\partial}{\partial J_{\vec{p}} }\bigg] e^{\mathcal W_s[J,\bar{J}]} =0\,,\label{socle2}
\end{align}
where we have used  the fact that, for all polynomial $P(T,\bar T)$ the following identity holds:
\bea
\int\, d\mu_C\,\, P(T,\bar T) e^{\langle \bar J, T\rangle+\langle\bar T, J\rangle}=\int\, d\mu_C\,\, P\Big(\frac{\partial}{\partial \bar J},\frac{\partial}{\partial J}\Big)  e^{\langle \bar J, T\rangle+\langle\bar T, J\rangle}.
\eea
 Equation \eqref{socle2} is satisfied for all $i$. Now, expanding each derivative, the partition function $\mathcal Z_s[J,\bar{J}]=:e^{\mathcal W_s[J,\bar{J}]}$ of the theory defined by the action \eqref{s4}  verify the following  (WT identity),
\bea\label{left}
\sum_{\vec{p}_{\bot_i}, \vec{q}_{\bot_i}} \delta_{\vec p_{\bot_i}\vec q_{\bot_i}}  \bigg\{\big[C_s(\vec{p}\,^{2})-C_s(\vec{q}\,^{2})\big]\left(\frac{\partial^2 W_s}{\partial \bar{J}_{\vec{q}\,}\,\partial {J}_{\vec{p}}}+\bar{M}_{\vec{p}}M_{\vec{q}\,}\right)-\bar{J} _{\vec{p}}\,M_{\vec{q}\,}+{J} _{\vec{q}\,}\bar{M}_{\vec{p}}\bigg\}=0\,.\label{Ward0}
\eea
WI-identity contains some informations on the relations between Green functions. In particular, they provide a relation between $4$ and $2$ points functions, which,  maybe translated as a relation between wave function renormalization $Z$ and vertex renormalization $Z_\lambda$. Applying $\partial^2/\partial M_{\vec{r}}\,\partial \bar{M}_{\vec{s}}$ on the left hand side of \eqref{left}, and taking into account the relations
\begin{equation}
\frac{\partial M_{\vec{p}}}{\partial {J}_{\vec{q}\,}}= \frac{\partial^2 \mathcal W_s}{\partial \bar{J}_{\vec{p}}\,\partial {J}_{\vec{q}\,}}\, \quad\mbox{and }\quad  \, \frac{\partial \Gamma_s}{\partial M_{\vec{p}}}=\bar{J}_{\vec{p}}-r_s(\vec{p})\bar{M}_{\vec{p}}\,,
\end{equation}
as well as the definition $G_{s\,,\vec{p}\vec{q}\,}^{-1}:=(\Gamma^{(2)}_s+r_s\big)_{\vec{p} \vec{q}\,}$, we find that
\begin{align}
\nonumber\sum_{\vec{p}_{\bot_i}, \vec{q}_{\bot_i}} \delta_{\vec p_{\bot_i}\vec q_{\bot_i}}  &\bigg[\big[C_s(\vec{p}\,^{2})-C_s(\vec{q}\,^{{2}})\big]\bigg[\frac{\partial^2 G_{s\,,\vec{p},\vec{q}\,}}{\partial M_{\vec{r}\,}\,\partial \bar{M}_{\vec{s}}}+\delta_{\vec{p}\vec{r}}\,\delta_{\vec{q}\,\vec{s}\,}\bigg]-\Gamma^{(2)}_{s\,,\vec{r}\vec{p}}\,\delta_{\vec{s}\,\vec{q}\,}+\Gamma^{(2)}_{s\,,\vec{s}\,\vec{q}\,}\delta_{\vec{p}\vec{r}}\\
&-r_s(\vec{p}\,^2)\delta_{\vec{r}\vec{p}}\,\delta_{\vec{s}\,\vec{q}\,}+r_s(\vec{q}^{\,2})\delta_{\vec{s}\,\vec{q}\,}\delta_{\vec{p}\vec{r}}-\Gamma_{s,\vec{r};\vec{s}\vec{p}}^{(1,2)}\,M_{\vec{q}\,}+\Gamma^{(2,1)}_{s,\vec{r}\vec{q};\vec{s}}\bar{M}_{\vec{p}}\bigg] =0\,,
\end{align}
and therefore the proposition \eqref{propnew} is well given.
\end{proof}
\noindent
 In the deep UV, for large  scale $s$, a continuous approximation for variables is suitable. Then, setting $r_1=p_1$, $\vec{p}\to \vec{q}$, $r_1\to s_1$, we get finally,   in the deep UV, the $4$ and $2$-point functions are related as (on both sides, $r_1=p_1$):
\begin{align}\label{WT-id}
\sum_{\vec{r}_{\bot_1}} G_s^2(\vec{r}\,)\frac{dC_s^{-1}}{dr_1^{2}}(\vec{r}\,)\Gamma_{s,\vec{r},\vec{r},\vec{p},\vec{p}}^{(4)}=\frac{d}{dp_1^{2}}\left(C_\infty^{-1}(\vec{p}\,)-\Gamma_s^{(2)}(\vec{p}\,)\right).
\end{align}
To give  more comment  on  the structure of this equation, we have to specify the structure of the vertex function. To this end, we use  this loop to discard the irrelevant contributions, and we keep only the \textit{melonic contribution} of the function $\Gamma^{(4)}$,  denoted by $\Gamma_{\text{melo}}^{(4)}$. In the symmetric phase, the melonic contribution $\Gamma_{\text{melo}}^{(4)}$ may be defined as the part of the function $\Gamma^{(4)}$ which decomposes as a sum of melonic diagrams in the perturbative expansion.
The structure of the melonic diagrams has been extensively discussed in the literature, and specifically for the approach that we propose here in \cite{Lahoche:2018oeo}-\cite{Lahoche:2018vun}. Formally, they are defined as the graphs optimizing the power counting; and they family can be build from the recursive definition of the vacuum melonic diagrams, from the cutting of some internal edges. Among there interesting properties, these construction imply the following statement:
\begin{proposition}\label{propmelons}
Let $\mathcal{G}_N$  be a $2N$-point 1PI melonic diagrams build with more than one vertices for a purely quartic melonic model. We call external vertices the vertices hooked to at least one external edge of $\mathcal{G}_N$ has :

\item $\bullet$  two external edges per external vertices, sharing $d-1$ external faces of length one. 
\item $\bullet$ $N$ external faces of the same color running through the interior of the diagram. 
\end{proposition}
As a direct consequence of the proposition \ref{propmelons}, we expect that melonic $4$-points functions is decomposed as:
\begin{equation}
\Gamma_{\text{melo}}^{(4)}=\sum_{i=1}^d \Gamma_{\text{melo}}^{(4),i}\,,
\end{equation}
the index $i$ running from $1$ to $d$ corresponding to the color of the $2$ internal faces running through the interiors of the diagrams building $ \Gamma_{\text{melo}}^{(4),i}$. Moreover the monocolored components have the following structure:
\begin{equation}
\Gamma_{\text{melo\,}\vec{p}_1,\vec{p}_2,\vec{p}_3,\vec{p}_4}^{(4),i} = \vcenter{\hbox{\includegraphics[scale=0.8]{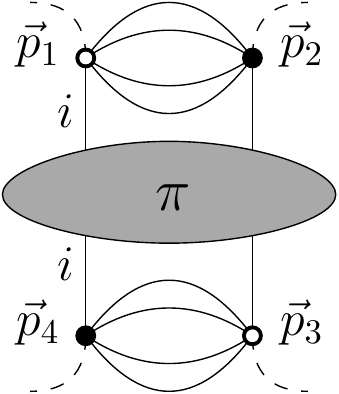} }}+ \vcenter{\hbox{\includegraphics[scale=0.8]{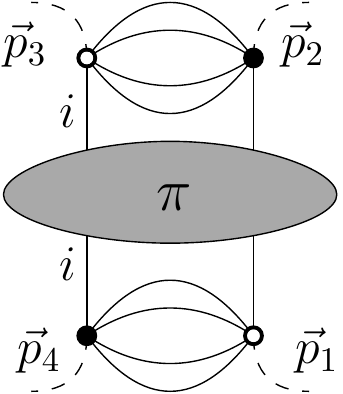} }}\,,\label{decomp4}
\end{equation}
the permutation of the external momenta $\vec{p}_1$ and $\vec{p}_3$ coming from Wick's theorem: There are four way to hook the external fields on the external vertices (two per type of field). Moreover, the simultaneous permutation of the black and white fields provides exactly the same diagram, and we count twice each configurations pictured on the previous equation. This additional factor $2$ is included in the definition of the matrix $\pi$, whose entries depend on the components $i$ of the external momenta running on the boundaries of the external faces of colors $i$, connecting together the end vertices of the diagrams building $\pi$.  \\

\noindent
Inserting \eqref{decomp4} into the Ward identity given from equation \eqref{WT-id}, we get some contributions on the left hand side, the only one relevant of them in the deep UV being, graphically:
\begin{equation}\label{diage}
\vcenter{\hbox{\includegraphics[scale=0.8]{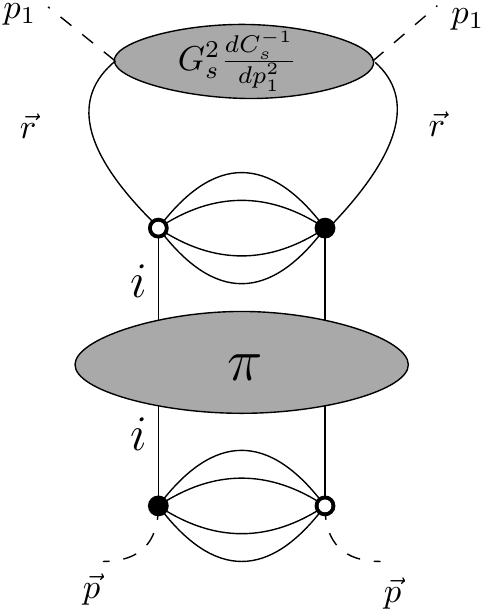} }}+\mathcal{O}\Big(\frac{1}{s}\Big)=\frac{d}{dp_1^{2}}\left(C_s^{-1}(\vec{p}\,)-\Gamma^{(2)}(\vec{p}\,)\right)\,.
\end{equation}
Setting $\vec{p}=\vec{0}$, and using the definition of $C_s^{-1}$ as well as the definition of $C_\infty^{-1}$, the right hand side is reduced to $Z_{-\infty}-Z$. Moreover, the diagram on the left hand side can be written with the following equation $Z_{-\infty} \mathcal{L}_s\, \pi_{00}$ such that the following equality holds:
\begin{equation}
Z_{-\infty} \mathcal{L}_s\, \pi_{00}=Z_{-\infty}-Z\,,
\end{equation}
where we have defined $Z_{-\infty} \mathcal{L}_s$ as:
\begin{equation}
Z_{-\infty} \mathcal{L}_s:=\sum_{\vec{p}\in\mathbb{Z}^{d}} \left(Z_{-\infty}+\frac{\partial r_s}{\partial p_1^{2}}(\vec{p}\,)\right)G^2_s(\vec{p}\,)\delta_{p_10}\,.
\end{equation}
Finally, from definition \eqref{decomp4} we expect that $\Gamma^{(4)}_{\text{melo},\vec{0},\vec{0},\vec{0},\vec{0}}=2\pi_{00}$, and because of the renormalization conditions \eqref{rencond} we must have the relation: $\pi_{00}=2\lambda_{41}(s)$. Therefore,
in the deep UV regime, the Ward identity between $4$ and $2$ point functions provides a non trivial relation between effective coupling and wave function renormalization:
\begin{equation}
2Z_{-\infty} \mathcal{L}_s\, \lambda_{41}=Z_{-\infty}-Z\,.\label{Wardutile}
\end{equation} 
\begin{remark}
Let us  give some important remarks regarding the derivation of the Ward identity \eqref{Wardutile}. First of all, the WI is totally disconnected from the approximation used to solve the non-perturbative Wetterich equation \eqref{Wetterich}. The Wetterich equation and Ward identity are both two functional results, deduced from the definition of the partition function, and have to be treated on the same footing. Their origins, moreover, are completely disconnected. One of them comes from the scale dependence of the model due to the regulator term, the second one comes from the symmetry violation of the action (including source terms) under the $U(N)^d$ group and the formal translation-invariance of the Lebesgue measure. Viewing the set $\mathcal{Z}_s$ has a continuous family of models, one can say that the Wetterich equation dictate how to move from $\mathcal{Z}_s$ to $\mathcal{Z}_{s+\delta s}$ whereas the WI are constraints between the observables at fixed  $s$. 
\end{remark}
From now, in the hope to provide the proof that $p_+$ does not live in the constraint line coming from Ward identity \eqref{Wardutile}, let us give the following result which will be prove in the next section.
\begin{proposition} \textbf{Structure equation for effective coupling:}
In the deep UV, the effective melonic coupling is given in terms of the renormalized coupling $\lambda_{41}^r$ and the renormalized effective loop $\bar{\mathcal{A}}_s:=\mathcal{A}_s-\mathcal{A}_{s=-\infty}$ as:
\begin{equation}
\lambda_{41}(s)=\frac{\lambda_{41}^r}{1+2\lambda_{41}^r\bar{\mathcal{A}}_s},\quad \dot{\lambda}_{41}=-2\lambda_{41}^2\,\dot{\mathcal{A}}_s\,.\label{structure}
\end{equation}
where we defined the quantity $\mathcal{A}_s$ as:
$
\mathcal{A}_s:=\sum_{\vec{p}\in\mathbb{Z}^{(d-1)}}\,G_s^2(\vec{p}\,)\,.
$
\end{proposition}
The constraint providing from the Ward identity, which relies the  $\beta$-functions and the anomalous dimension is given by: 
\bea\label{contrainte1}
\mathcal{C}(\bar{\lambda},\bar{m}^2):=\beta_{{41}}+\eta\bar\lambda_{41}\Big(1-\frac{\bar\lambda_{41}\pi^2}{(1+\bar m^2)^2}\Big)-\frac{2\bar{\lambda}_{41}^2\pi^2}{(1+\bar m^2)^3}\beta_m=0
\eea
This relation need to be taken into account in the Wetterich flow equation and therefore in the search of fixed point. To prove this relation, let us consider the
 derivative of $Z$ with respect to $s$ using expression  \eqref{Wardutile} and \eqref{structure}:
\begin{equation}\label{derivative1}
\dot{Z}=(Z_{-\infty}-2\lambda_{41} Z_{-\infty}\mathcal{L}_s)\frac{ \dot{\lambda}_{41}}{\lambda_{41}}-2Z_{-\infty} \dot{\Delta}_s\,\lambda_{41}.
\end{equation}
In the above relation  we have  used the  decomposition of  $\mathcal{L}_s=\mathcal{A}_s+\Delta_s$.  Remark that the Ward identity \eqref{Wardutile} can be written as 
$2\lambda_{41}\mathcal{L}_s=1-\bar Z$ where $\bar Z=Z/Z_{-\infty}$.
Then \eqref{derivative1} becomes:
\beq\label{ddddx}
\frac{\dot{Z}}{Z}=\frac{\dot{\lambda}_{41}}{\lambda_{41}}-2\frac{Z_{-\infty}}{Z}\dot{\Delta}_s\lambda_{41}.
\eeq
We now use the dimensionless quantities $\bar m$, $\bar \lambda_{41}$, $\bar B_s$ such that $\Delta_s=\frac{\bar Z}{Z^2}\bar B_s$ and reexpressing \eqref{ddddx} as:
\bea
\beta_{41}=-\eta\bar\lambda_{41}+2\bar\lambda_{41}(-\eta \bar{B}_s+\dot{\bar{B}}_s)
\eea
where $\bar{B}_s$ and $\dot{\bar{B}}_s$ much be simply compute using the integral representation of the sum. We come to:
\bea
\bar{B}_s=-\frac{\pi^2}{2(1+\bar m^2)^2},\quad \dot{\bar{B}}_s=\frac{\pi^2\beta_m}{(1+\bar m^2)^3},
\eea
and therefore \eqref{contrainte1} is well given. It is time to prove that this constraint violate the existence of the fixed point $p_+$. Let $p$ is a arbitrary fixed point of the theory. We get $\beta_m(p)=0=\beta_{41}(p)=0.$ Then the constraint \eqref{contrainte1} implies that at the point $p$ we get 
\bea
\eta\bar\lambda_{41}\Big(1-\frac{\bar\lambda_{41}\pi^2}{(1+\bar m^2)^2}\Big)(p)=0.
\eea
The particular solution $\bar\lambda_{41}=0$ correspond to the Gaussian fixed point. For $\bar\lambda_{41}\neq 0$ we have only
\bea\label{contnewt}
\eta=0,\mbox{ or } \frac{\bar\lambda_{41}\pi^2}{(1+\bar m^2)^2}=1.
\eea
It is clear that the fixed point $p_+=(-0.55, 0.0025)$, $\eta\approx 0.7$ violate these constraints i.e. does not satisfied the contraint equation \eqref{contnewt}. The same conclusion can be made for all choice of the regulator see \cite{Lahoche:2018ggd}. Finally it is possible to improve the truncation by using the so called effective vertex expansion. In this case, the fixed point obtained by solving the flow equation also violate the Ward constraint \eqref{contnewt}. We will study this point in the next section.
\section{Effective vertex expansion method for the melonic sector}\label{sec4}
The effective vertex-expansion described in \cite{Lahoche:2018ggd}-\cite{Lahoche:2018vun} allows to establish the structure of the Feynman graphs of our models and leads to the structure  equations in the leading order sector.  It can help to establish the flow equations without truncation. The Feynman graphs of the colored tensor model are $(d+1)$-colored graphs \cite{Gurau:2011xq}-\cite{Gurau:2013pca}. For the sake
of completeness, we remind here a  few facts about these graphs,
their representation as stranded graphs and their uncolored version. The graphs that we consider possibly bear external edges,
that is to say half-edges hooked to a unique vertex. We denote
$\cG$ a colored graph, $\cL(\cG)$ 
the set of its internal edges ($L(\cG)=|\cL(\cG)|$). A colored graph is said closed if it has no external edges and open otherwise.
  Let $\cG$ be a $(d+1)$-colored graph and $S$ a subset of
  $\{0,\dots,d\}$. We note $\cG^{S}$ the spanning subgraph of $\cG$  induced by the edges of colors in $S$.
 Then for all $0\leq i,j\leq  d$, $i\neq j$, a face of colors $i,j$ is a connected component of $\cG^{\{i,j\}}$. A face is open (or external) if it contains an external edge and
closed (or internal) otherwise. The set of closed faces of a graph $\cG$ is
written $\cF(\cG)$ ($F(\cG)=|\cF(\cG)|$). The structure of the boundary graph of $\cG$ denoted by $\partial\cG$ will be useful in the construction of the leading order contribution which may be considered  in the derivative expansion to compute the structure equations and therefore  the flow equations.
\begin{definition}
Consider $\mathcal{G}$  as a connected Feynman graph with $2N$ external
edges. The boundary graph $\partial\mathcal{G}$ is obtained from
$\mathcal{G}$ keeping only the external blacks and whites nodes
hooked to the external edges, connected together with colored
edges following the path drawn from the boundaries of the
external faces in the interior of the graph $\mathcal{G}$.
$\partial\mathcal{G}$ is then a tensorial invariant itself with
$N$ blacks (resp. whites) nodes. An illustration is given on
Figure \eqref{fig2}.
\end{definition}
\begin{center}
$\vcenter{\hbox{\includegraphics[scale=0.7]{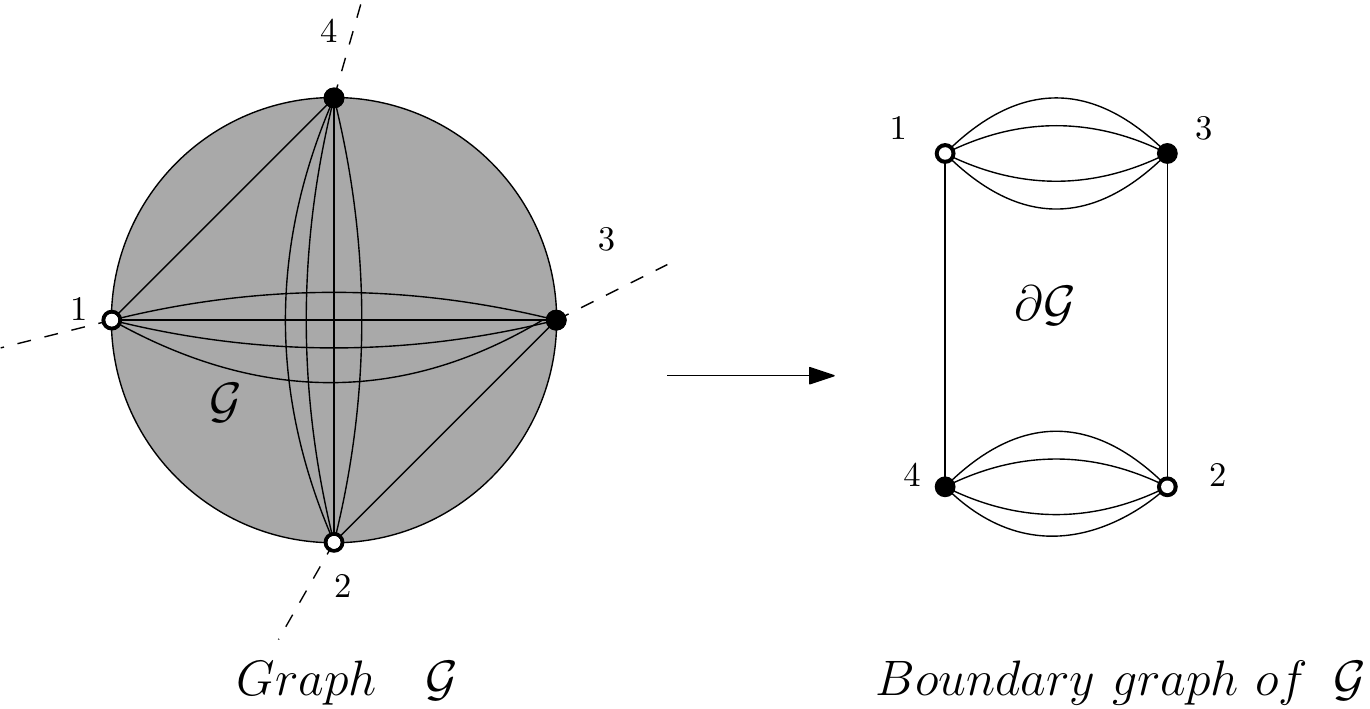} }} $
\captionof{figure}{An opening Feynman graph with $4$ external
edge and its boundary graph. The strand in the interior of
$\mathcal{G}$ represent the path following by the external
faces.}\label{fig2}
\end{center}

 The power counting theorem of  these models show that the divergence degree of arbitrary Feynman graph $\cG$ is
\bea
\omega(\cG)=-2L(\cG)+F(\cG)
\eea
The topological operation on the edge of the graph $\cG$ such as contraction is studied extensively in a lot of literatures. We let the reader consult  \cite{Gurau:2011xq}-\cite{Gurau:2013pca} and references therein. This operation plays 
  an important role in the power counting theorem and allowed to identify the structure of the graph. It makes the connection between the divergence degree of $\cG$ and the spanning tree denoted by $\mathcal T$.  Let "$\setminus$" is the operation of contraction, we get the following proposition:
\begin{proposition}
Under the contraction of the spanning tree edge the number of internal faces is invariant i.e.
$
F(\cG)=F(\cG\setminus\mathcal T).
$
The graph $\cG\setminus\mathcal T$ is called the rosette.
\end{proposition}
Note that the contraction of the edge $e\in\mathcal{L}(\cG)$, which leads  to the corresponding graph $\bar\cG=\cG\setminus \{e\}$ is such that
$
\omega(\cG)=\omega(\bar\cG)-2(V-1),
$
and   the divergent degree of the rosette can be easly computed,  using the following formula corresponding to the contraction of $k$-dipole: $\omega(\cG)=-2L+k(L-V+1).$  Then
 the arbitrary Feynman graph $\cG$ is melonic if its boundary graph have the elementary melon structure i.e.  the number of face is maximal:
\bea
F(\cG_{melon})=(d-1)(L-V+1).
\eea
Due to the existence of the $1/N$-expansion of tensors models ($N$ denoting the size of the tensor) which provides in return a topological expansion of the partition function in terms of the generalization of genus called Gurau number $\varphi$,
does not yield a topological expansion but rather 
a combinatorial expansion in terms of the degree
of the graph.  For a colored closed graph $\cG$, the degree $\varpi(\cG)$ is such that
 for the melon $\varpi(\cG_{melon})=0$. 

\subsection{Structure equations and compactibility with Ward-identities}
The Structure equations  is the relations between correlation function and allows to establish a constraint between  $\beta$-functions for mass,   interactions couplings and wave function renormalization. These relations are obtained in the deep UV limit (i.e. in the domain $1\ll e^s\ll \Lambda$) without any assumption about the $\beta$-functions and without any truncation of the effective action $\Gamma_s$. The only assumption concern the choice of the initial conditions, ensuring the perturbative consistency of the full partition function.
The first structure equation concern the self energy (or 1PI $2$-point functions). It takes  place as the \textit{closed equation} for self energy. \footnote{The rank of the tensors is fixed to $5$, and we denote it by $d$ to clarify the proof(s).} Let us summarize in the following proposition
\begin{proposition}
In the melonic sector, the self energy $\Sigma_s(\vec{p}\,)$ is given by the \textit{closed equation} which takes into account the  effective coupling $\lambda_{41}(s)$ as:
\begin{equation}\label{eq2points}
-\Sigma_s(\vec{p}\,)=2\lambda_{41}^rZ_\lambda\sum_{\vec{q}} \left(\sum_{i=1}^d \delta_{p_iq_i}\right)G_s(\vec{q}\,)\,.
\end{equation}
In the same way,  in the melonic sector, the perturbative zero-momenta 1PI four-point contribution $\Gamma^{(4),i}_{s,\vec 0\vec 0;\vec 0\vec 0}$ is given by:
\begin{equation}
\Gamma^{(4),i}_{s,\vec 0\vec 0;\vec 0\vec 0}=2\pi_{00}=\frac{4Z_\lambda\lambda_{41}^r}{1+2\lambda_{41}^r Z_\lambda \mathcal A_s}\,,
\end{equation}
where $\mathcal A_s$ is defined as:
\begin{equation}
\mathcal A_s=\sum_{\vec p_{\bot}}[G_s(\vec p_{\bot})]^2\,,\,\vec{p}_{\bot} := (0,p_1,\cdots,p_d)\,,
\end{equation}
$G_s(\vec{p})$ being the effective propagator :  $G_s^{-1}(\vec p\,)=Z_{-\infty}\vec p\,^2+m^2+r_s(\vec p\,)-\Sigma_s(\vec p\,)\,.$ Let us recall that $Z_{-\infty}$ and $m_0$ are the counter-terms discarding the UV divergences of the original partition function, the initial conditions in the UV are given such that the classical action contain only renormalizable interactions. \\
\end{proposition}

\noindent
\begin{proof} Concerning the proof of relation \eqref{eq2points} we let the reader to consult the  reference \cite{Samary:2014tja}.
 Let us define $4Z_\lambda\lambda_{41}^r\Pi$ as the zero momenta melonic $4$-points functions made into  the graphs for which two vertices maybe singularized (i.e. by graphs which are at least of order $2$ in the perturbative expansion). We have\footnote{The notations are similar to the ones used for the previous proof. The context however allows to exclude any confusion.}:
\begin{equation}
2\pi_{00}=:4Z_\lambda\lambda_{41}^r(1+\Pi).
\end{equation}
Because of the face connectivity of the melonic diagrams, the boundary vertices may be such that the two internal faces of the same color running on the interior of the diagrams building $\Pi$ pass through of them. Then  we have the following structure:
\begin{equation}
-4Z_\lambda\lambda_{41}^r\Pi=\vcenter{\hbox{\includegraphics[scale=0.7]{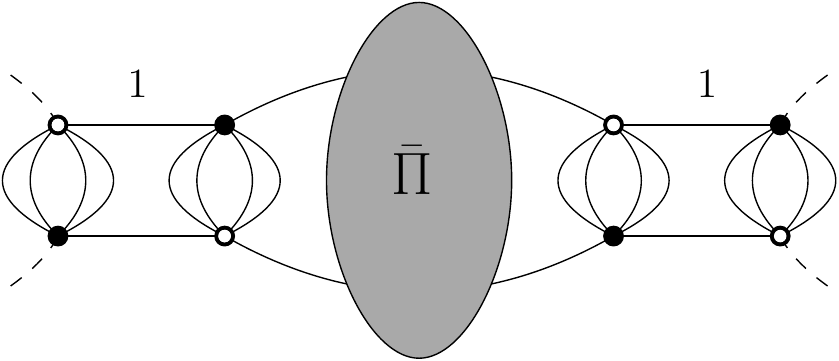} }}, 
\end{equation}
where the grey disk is a sum of Feynman graphs. Note that it is the only configuration of the external vertices in agreement with the assumption that $\Pi$ is building with the melonic diagrams. Any other configurations of the external vertices are not melonics. At the lowest order, the grey disk corresponds to propagator lines, 
\begin{equation}
-4Z_\lambda\lambda_{41}^r\Pi^{(2)}=8Z_\lambda^2 (\lambda_{41}^r)^2 \mathcal{A}_s\vert_{\lambda_{41}^r=0}\equiv\vcenter{\hbox{\includegraphics[scale=0.7]{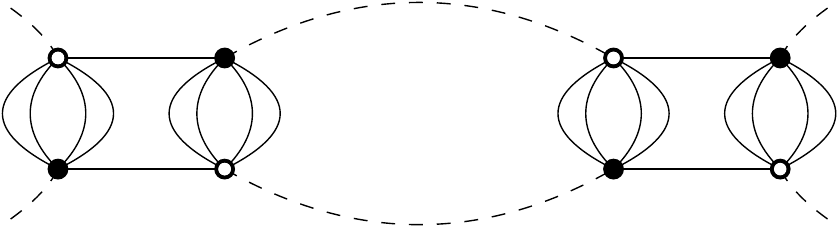} }} \,.
\end{equation}
Note that,the external faces have the same color. Now, we can extract the amputated component of $\bar{\Pi}$, say $\bar{\Pi}^{\prime}$ (which contains at least one vertex, and is irreducible by hypothesis) extracting the effective melonic propagators connected to the  dotted lines linked to $\bar{\Pi}$. We get:
\begin{equation}
-4Z_\lambda\lambda_{41}^r\Pi=\vcenter{\hbox{\includegraphics[scale=0.7]{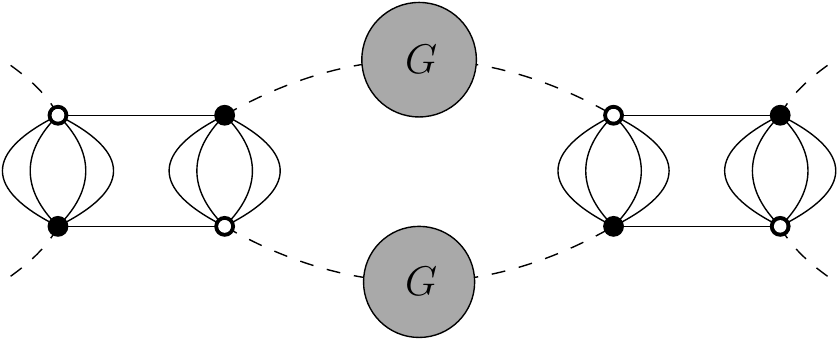} }} +\vcenter{\hbox{\includegraphics[scale=0.7]{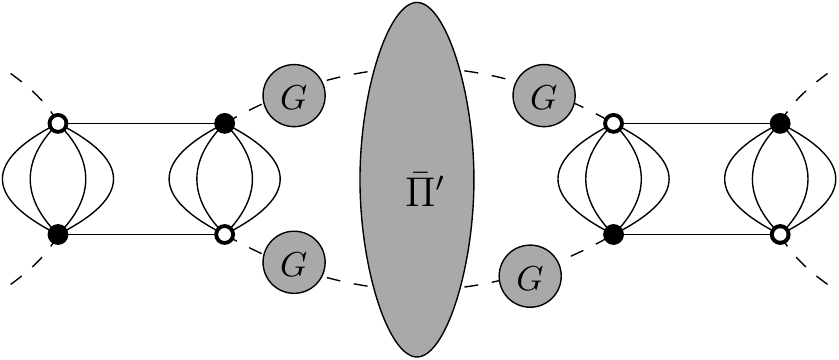} }}\,.
\end{equation}
At first order, $\bar{\Pi}^{\prime}$ is built with a single vertex, and there are only one configuration in agreement with the melonic structure, i.e. maximazing the number of internal faces. The higher order contributions contain at least two vertices, and the argument may be repeated so that the function $\bar{\Pi}^{\prime}$ appears. Finally we deduce the closed relation:
\begin{equation}
\vcenter{\hbox{\includegraphics[scale=0.5]{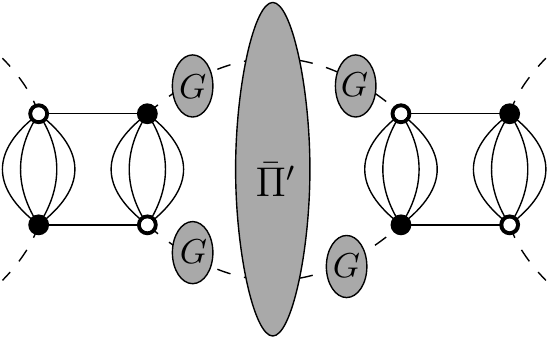} }}=\vcenter{\hbox{\includegraphics[scale=0.5]{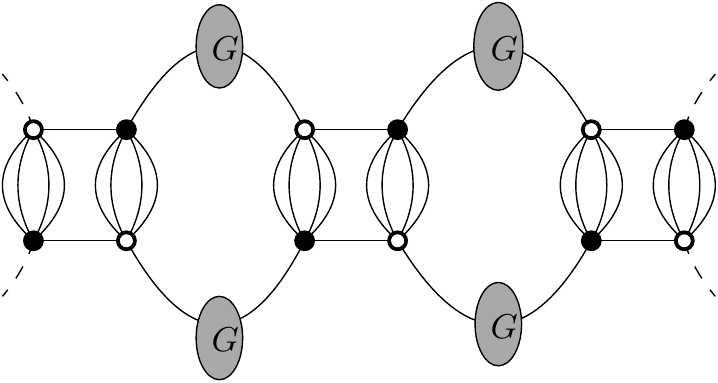} }}+\vcenter{\hbox{\includegraphics[scale=0.5]{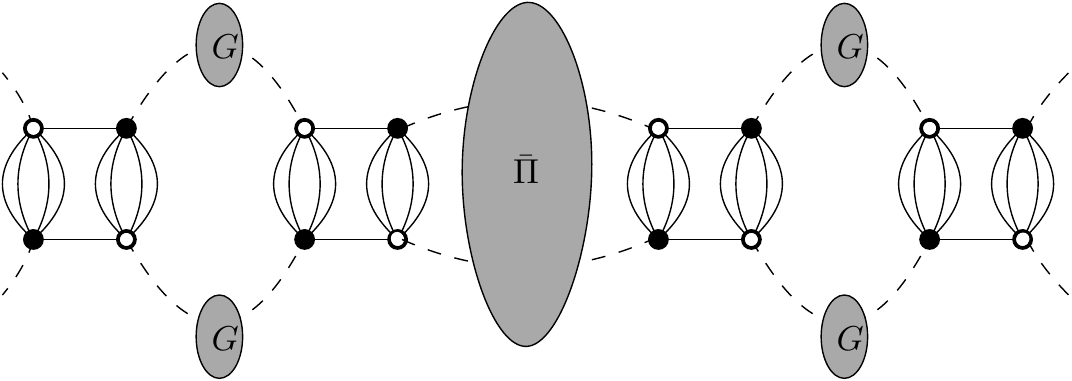} }}\,. \label{recursion1}
\end{equation}
This equation can be solved recursively as an infinite sum
\begin{equation}
-4Z_\lambda\lambda_{41}^r\Pi=\vcenter{\hbox{\includegraphics[scale=0.6]{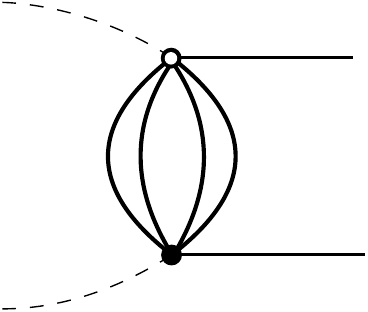} }} \left\{\sum_{n=1}^\infty \left(\vcenter{\hbox{\includegraphics[scale=0.6]{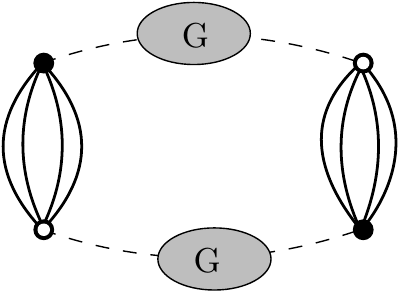} }}\right)^n \right\} \vcenter{\hbox{\includegraphics[scale=0.6]{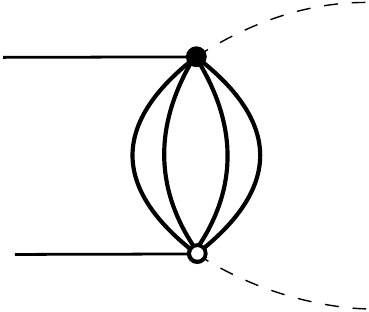} }}\,,
\end{equation}
which can be formally solved as
\begin{equation}
2\pi_{00}= 4Z_\lambda\lambda_{41}^r\left(1-\vcenter{\hbox{\includegraphics[scale=0.5]{pimiddle.pdf} }}\right)^{-1}\,.
\end{equation}
The loop diagram $\vcenter{\hbox{\includegraphics[scale=0.3]{pimiddle.pdf} }}$ maybe easily computed recursively from the  definition of melonic diagrams, or directly using Wick theorem  for a one-loop computation with the effective propagator $G$. The result is:
\begin{equation}
\vcenter{\hbox{\includegraphics[scale=0.5]{pimiddle.pdf} }}=-2Z_\lambda \lambda_{41}^r \mathcal{A}_s\,,
\end{equation}
and the proposition is proved. \\
\end{proof}
Note that this construction can be  easily cheeked to be compatible with Ward identity, especially in the form \eqref{diage}. Conversely, the last result may be derived directly from the equation \eqref{diage} and from the closed equation for the $2$-point function \eqref{eq2points} (see \cite{Lahoche:2018vun}). To prove these two results we only assume that the classical mean field vanish, we deduce from our previous proof, essentially based on the assumption that the effective vertices are analytic with respect to the renormalized coupling, that  the analytic domain cover what we called symmetric phase. 
In the hope to  extract the expression of the counter-terms at all orders and  to show that the wave function renormalization and the $4$-points vertex renormalization are the same. We have the following result:
\begin{proposition}\label{propcounterterms}
Choosing the following renormalization prescription:
\begin{equation}
\Gamma^{(4),1}_{s=-\infty,\vec 0\vec 0;\vec 0\vec 0}=4\lambda_{41}^r\,\,;\,\, \Gamma^{(2)}_{s=-\infty}(\vec{p}\,)=m_r^2+\vec{p}\,^2+\mathcal{O}(\vec{p}\,^2)\,,
\end{equation}
where $m_r^2$ and $\lambda_{41}^r$ are the renormalized mass and coupling constant; the counter-terms are given by:
\begin{equation}
Z_\lambda=\frac{1}{1-2\lambda_{41}^r\mathcal A_{s=-\infty}},\,\,;\,\, Z_{-\infty}=Z_\lambda \,\,;\,\, m^2=m_r^2+\Sigma_{s=-\infty}(\vec{p}=0)\,,
\end{equation}
where $\Sigma_s$ denote the melonic self-energy. 
\end{proposition}
\begin{proof} From Proposition \ref{prop1}, we can get:
\begin{equation}
\Gamma^{(4),i}_{s,\vec 0\vec 0;\vec 0\vec 0}=\frac{4Z_\lambda\lambda_{41}^r}{1+2\lambda_{41}^r Z_\lambda \mathcal A_s}=\frac{4\lambda_{41}^r}{Z_\lambda^{-1}+2\lambda_{41}^r \mathcal A_s}\,.
\end{equation}
Then, setting $s=-\infty$, we deduce that
\begin{equation}
Z_\lambda^{-1}+2\lambda_{41}^r \mathcal A_{-\infty}=1\to Z_\lambda=\frac{1}{1-2\lambda_{41}^r\mathcal A_{-\infty}}\,.\label{explicitZlambda}
\end{equation}
We now concentrated our self  on to $Z_{-\infty}$ and $m^2$. Without lost of generality, the inverse of the effective propagator $\Gamma^{(2)}_s$ has the following structure:
\begin{align}
\Gamma_{s=-\infty}^{(2)}(\vec p\,)&=Z_{-\infty}\vec p\,^2 +m^2-\Sigma_{s=-\infty}(\vec p)\\
&=Z_{-\infty}\vec p\,^2 +m^2-\Sigma_{s=-\infty}(\vec 0) -\vec p\,^2\Sigma^\prime_{{s=-\infty}}(\vec 0)+\mathcal{O}(\vec{p}\,^2)\\
&=(Z_{-\infty}-\Sigma_{s=-\infty}^\prime(0))\vec p\,^2+m^2-\Sigma_{s=-\infty}(\vec 0)+\mathcal{O}(\vec{p}\,^2)
\end{align}
with the notation: $\Sigma^\prime(\vec{0}):= \partial \Sigma/\partial p_1^2(\vec{p}=\vec{0}\,)$. Then  from the renormalization conditions, we have :
\begin{equation}
Z_{-\infty}-\Sigma_{s=-\infty}^\prime(0)=1\,\,,\,\, m^2-\Sigma_{s=-\infty}(\vec 0)=m_r^2\,.
\end{equation}
Setting $s=-\infty$ in the closed equation for the $2$-point correlation function, and by deriving  with respect to $p_1$ for $\vec{p}=\vec{0}$, we get:
\begin{equation}
1-Z_{-\infty}=-2\lambda_{41}^r Z_\lambda  \mathcal{A}_{s=-\infty}\,.
\end{equation}
Using the explicit expression for $Z_\lambda$ in  \eqref{explicitZlambda}, we get finally:
\begin{equation}
(1-Z_{-\infty})(1-2\lambda_{41}^r  \mathcal{A}_{s=-\infty} )=-2\lambda_{41}^r   \mathcal{A}_{s=-\infty} \,\,\to\,\, Z_{-\infty}=Z_\lambda\,.
\end{equation}
\end{proof}
Now, consider the monocolor $4$-points function $\Gamma^{(4),i}_{s,\vec 0\vec 0;\vec 0\vec 0}$. If we replace $Z_\lambda$ by its expression from Proposition \ref{propcounterterms}, we deduce that
\begin{equation}
\Gamma^{(4),i}_{s,\vec 0\vec 0;\vec 0\vec 0}=\frac{4\lambda_{41}^r}{1+2\lambda_{41}^r  \bar{\mathcal A}_s}\,,
\end{equation}
with the definition: $\bar{\mathcal A}_s:= \mathcal A_s-\mathcal A_{s=-\infty}$. In other words, we have an explicit expression for the effective coupling $\lambda_{41}(s):=\frac{1}{4} \Gamma^{(4),i}_{s,\vec 0\vec 0;\vec 0\vec 0}$,
\begin{equation}
\lambda_{41}(s)=\frac{\lambda_{41}^r}{1+2\lambda_{41}^r  \bar{\mathcal A}_s}\,,\label{effectivecoupling}
\end{equation}
from which we get
\begin{equation}
\partial_s\lambda_{41}(s)=-\frac{2(\lambda_{41}^r)^2\dot{\mathcal A}_s}{(1+2\lambda_{41}^r\Delta \mathcal A_s)^2}=-2\lambda_{41}^2(s)\dot{\mathcal A}_s\,.
\end{equation}
In the above relation  we introduce the dot notation $\dot{\mathcal{A}}_s=\partial_s{\mathcal{A}}_s$
\begin{equation}
\mathcal A_s=\sum_{\vec p_{\bot}}\frac{1}{[\Gamma_s^{(2)}(\vec p_{\bot})+r_s(\vec p_{\bot})]^2},\quad 
\dot{\mathcal A}_s=-2\sum_{\vec p_{\bot}}\frac{\dot{\Gamma}_s^{(2)}(\vec p_{\bot})+\dot{r}_s(\vec p_{\bot})}{[\Gamma_s^{(2)}(\vec p_{\bot})+r_s(\vec p_{\bot})]^3}.
\end{equation}
In proposition \ref{propcounterterms} we have investigated the relations between counter-terms i.e. we have considered the melonic equations as Ward identities for $s=-\infty$. Far from the initial conditions, the Taylor expansion of the $2$-point function $\Gamma_s^{(2)}(\vec{p}\,)$ is written  as:
\begin{equation}
\Gamma_s^{(2)}(\vec p\,)=m_r^2+(\Sigma_s(\vec 0\,)-\Sigma_0(\vec 0\,))+(Z_{-\infty}-\Sigma_s'(\vec 0))\vec p\,^2+\mathcal{O}(\vec{p}\,^2)\,.
\end{equation}
We call the "physical" or \textit{effective} mass parameter $m^2(s)$ the first term in the above relation:
\begin{equation}
m^2(s):=m_r^2+(\Sigma_s(\vec 0\,)-\Sigma_0(\vec 0\,)),\,
\end{equation}
while the coefficient $Z_{-\infty}-\Sigma_s'(\vec 0)$ is the effective wave function renormalization and is denoted by $Z(s)$ i.e.
\begin{equation}
Z(s):=Z_{-\infty}-\Sigma_s'(\vec 0)\,.
\end{equation}
Now let us consider the closed equation given in proposition \ref{eq2points}. By  deriving  with respect to $p_1$ and by taking $\vec{p}=\vec{0}$, we get:
\begin{equation}
Z-Z_{-\infty}=-2\lambda_{41}^r Z_\lambda \sum_{\vec{p}_\bot} G^2_s(\vec{p}_\bot)(Z+r'_s(\vec{p}_\bot))\,.
\end{equation}
Using equation \eqref{effectivecoupling}, we can express $\lambda_{41}^r Z_\lambda$ in terms of the effective coupling $\lambda_{41}(s)$, and we get:
\begin{equation}
(Z-Z_{-\infty})(1-2\lambda_{41}(s)\mathcal{A}_s)=-2\lambda_{41}(s)\left(Z\mathcal{A}_s+\sum_{\vec{p}_\bot} G^2_s(\vec{p}_\bot)r'_s(\vec{p}_\bot)\right)\,,
\end{equation}
Then we come to the following relation
\begin{equation}
Z=Z_{-\infty}\left(1-2\lambda_{41}(s)\mathcal{L}_s\right)\,.\label{eqZ}
\end{equation}
At this stage, without all confusion let us clarify that: $Z_{-\infty}$ is the wave function counter-term i.e, whose divergent parts cancels the loop divergences, and whose finite part depend on the renormalization prescription. $Z(s)$ however is  fixing to be $1$ for $s=-\infty$ from our renormalization conditions. \\
\subsection{Flow equation from EVE method}
There are different methods to improve the crude truncations in the FRG literature. However, their applications for TGFTs remains difficult due to the non-locality of the interactions over the group manifold on which the fields are defined. A step to go out of the truncation method was done recently in \cite{Lahoche:2018oeo}-\cite{Lahoche:2018vun} with the \textit{effective vertex expansion} (EVE) method. Basically, the strategy is to close the infinite tower of equations coming from the exact flow equation, instead of crudely truncate them. To say more, the strategy is to complete the structure equation \eqref{structure} with a structure equation for $\Gamma^{(6)}$, expressing it in terms of the marginal coupling $\lambda$ and the effective propagator $G_s$ only. In this way, the flow equations around marginal couplings are completely closed. Note that this approach cross the first hypothesis motivating the truncation: We expect that so far from the deep UV, only the marginal interactions survive, and drag the complete RG flow. Moreover, any fixed point of the autonomous set of resulting equations are automatically fixed points for any higher effective melonic vertices building from effective quartic interactions. Finally, a strong improvement of this method with respect to the truncation method, already pointed out in \cite{Lahoche:2018oeo}-\cite{Lahoche:2018vun} is that it allows to keep the complete momenta dependence of the effective vertex. This dependence generate a new term on the right hand side of the equation for $\dot{Z}$, moving the critical line from its truncation's position. \\

\noindent
Let us consider the flow equation for $\dot{\Gamma}^{(2)}$, obtained from \eqref{Wetterich} deriving with respect to $M$ and $\bar{M}$:
\begin{equation}
\dot{\Gamma}^{(2)}(\vec{p}\,)=-\sum_{\vec{q}} \Gamma_{\vec{p},\vec{p},\vec{q},\vec{q}}^{(4)}\,G_s^2(\vec{q}\,) \dot{r}_s(\vec{q}\,)\,,\label{gamma2}
\end{equation}
where we discard all the odd contributions, vanishing in the symmetric phase. Deriving on both sides with respect to $p_1^{2}$, and setting $\vec{p}=\vec{0}$, we get:
\begin{equation}
\dot{Z}=-\sum_{\vec{q}} \Gamma_{\vec{0},\vec{0},\vec{q},\vec{q}}^{(4)\,\prime}\,G_s^2(\vec{q}\,)  \dot{r}_s(\vec{q}\,)-\Gamma_{\vec{0},\vec{0},\vec{q},\vec{q}}^{(4)}\,G_s^2(\vec{q}\,) \dot{r}_s(\vec{q}\,)\,,
\end{equation}
where the "prime" designates the partial derivative with respect to $p_1^{2}$. In the deep UV ($k\gg1$) the argument used in the $T^4$-truncation to discard non-melonic contributions holds, and we keep only the melonic diagrams as well. Moreover, to capture the momentum dependence of the effective melonic vertex $\Gamma_{\text{melo}}^{(4)}$ and compute the derivative $\Gamma_{\text{melo}\,,\vec{0},\vec{0},\vec{q},\vec{q}}^{(4)\,\prime}$, the knowledge of $\pi_{pp}$ is required. It can be deduced from the same strategy as for the derivation of the structure equation \eqref{structure}, up to the replacement :
\begin{equation}
\mathcal{A}_s \to \mathcal{A}_s(p):=\sum_{\vec{p}\in\mathbb{Z}^d}\,G^2_s(\vec{p}\,)\delta_{p_1p}\,,
\end{equation}
from which we get:
\begin{equation}
\pi_{pp}=\frac{2\lambda_{41}^r}{1+2\lambda_{41}^r\bar{\mathcal{A}}_s(p)}\,,\quad \bar{\mathcal{A}}_s(p):={\mathcal{A}}_s(p)-{\mathcal{A}}_{-\infty}(0)\,.
\end{equation}
The derivative with respect to $p_1^{2}$ may be easily performed, and from the renormalization condition \eqref{rencond}, we obtain: 
\begin{equation}
\pi_{00}^\prime=-4\lambda_{41}^2(s)\,\mathcal{A}_s^{\prime}\,,
\end{equation}
and the leading order flow equation for $\dot{Z}$ becomes:
\begin{equation}
\dot{Z}=4\lambda_{41}^2 \mathcal{A}_s^{\prime}(0) \,I_2(0)-2\lambda_{41} I_2^\prime(0)\,.
\end{equation}
As announced, a new term appears with respect to the truncated version \eqref{flownew}, which contains a dependence on $\eta$ and then move the critical line. The flow equation for mass may be obtained from \eqref{gamma2} setting $\vec{p}=\vec{0}$ on both sides. Finally, the flow equation for the marginal coupling $\lambda_{41}$ may be obtained from the equation \eqref{Wetterich} deriving it twice with respect to each mean field $M$ and $\bar{M}$. As explained before, it involves $\Gamma^{(6)}_{\text{melo}}$ at leading order, and to close the hierarchy, we use  the marginal coupling as a driving parameter, and express it in terms of $\Gamma_{\text{melo}}^{(4)}$ and $\Gamma^{(2)}_{\text{melo}}$ only. One again, from proposition \ref{propmelons}, $\Gamma^{(6)}_{\text{melo}}$ have to be split into $d$ monocolored components $\Gamma^{(6)\,,i}_{\text{melo}}$:
\begin{equation}
\Gamma^{(6)}_{\text{melo}}=\sum_{i=1}^d \Gamma^{(6)\,,i}_{\text{melo}}\,. 
\end{equation}
The structure equation for $ \Gamma^{(6)\,,i}_{\text{melo}}$ may be deduced following the same strategy as for  $\Gamma^{(4)\,,i}_{\text{melo}}$, from proposition \eqref{propmelons}. Starting from a vacuum diagram, a leading order $4$-point graph may be obtained opening successively two internal tadpole edges, both on the boundary of a common internal face. This internal face corresponds, for the resulting $4$-point diagram to the two external faces of the same colors running through the interior of the diagram. In the same way, a leading order $6$-point graph may be obtained cutting another tadpole edge on this resulting graph, once again on the boundary of one of these two external faces.
The reason this works is that,  in this may, the number of discarded internal faces is optimal, as well as the power counting. From this construction, it is not hard to see that the zero-momenta $ \Gamma^{(6)\,,i}_{\text{melo}}$ vertex function must have the following structure (see \cite{Lahoche:2018oeo}-\cite{Lahoche:2018vun} for more details):
\begin{equation}
 \Gamma^{(6)\,,i}_{\text{melo}}=(3!)^2\,\left(\vcenter{\hbox{\includegraphics[scale=0.5]{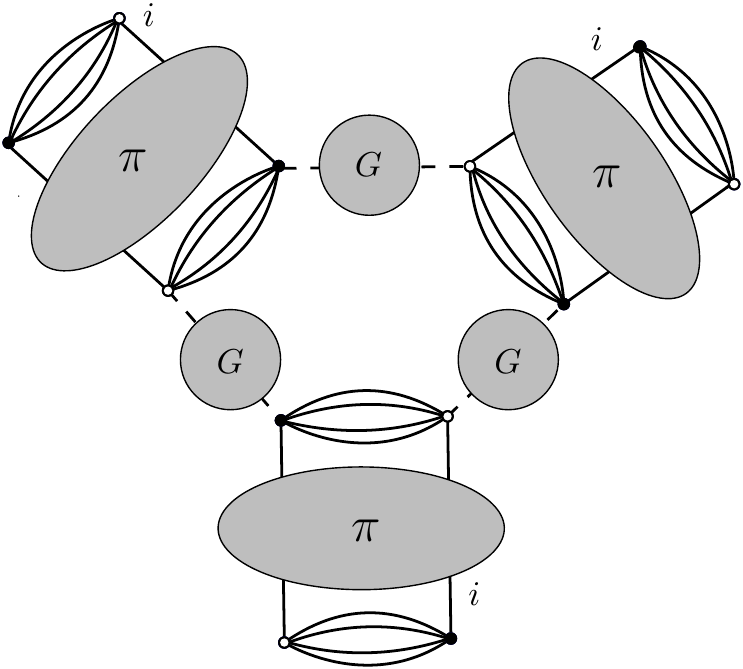} }}\right)\,,
\end{equation}
the combinatorial factor $(3!)^2$ coming from permutation of external edges. Translating the diagram into equation, and taking into account symmetry factors, we get:
\begin{equation}
 \Gamma^{(6)\,,i}_{\text{melo}} = 48Z^3(s)\bar{\lambda}_{41}^3(s)e^{-2 s} \bar{\mathcal{A}}_{2s}\,,
\end{equation}
with:
\begin{equation}
\bar{\mathcal{A}}_{2s}:= Z^{-3}e^{2 s} \sum_{\vec{p}\in\mathbb{Z}^{d-1}} G_s^3(\vec{p}\,)\,.
\end{equation} 
Note that this structure equation may be deduced directly from Ward identities, as pointed-out in \cite{Lahoche:2018vun} and \cite{Samary:2014tja}. The equation closing the hierarchy is then compatible with the constraint coming from unitary invariance. The flow equations involve now some new contributions depending on two sums, $\bar{\mathcal{A}}_{2s}$ and $\bar{\mathcal{A}}_s^{\prime}$, defined without regulation function $\dot{r}_s$. However, they are both power-counting convergent in the UV, and the renormalizability theorem ensures their finitness for all orders in the perturbation theory. For this reason, they becomes independent from the initial conditions at scale $\Lambda$ for $\Lambda\to\infty$; and as pointed out in \cite{Lahoche:2018vun},we get, using the Litim's regulator:
\begin{equation}
\bar{\mathcal{A}}_{2s}=\frac{1}{2}\frac{\pi^2}{1+\bar{m}^{2}}\left[\frac{1}{(1+\bar{m}^{2})^2}+\left(1+\frac{1}{1+\bar{m}^{2}}\right)\right]\,,
\end{equation}
and
\begin{equation}
\bar{\mathcal{A}}_s^{\prime}=\frac{1}{2}\pi^2\frac{1}{1+\bar{m}^{2}}\left(1+\frac{1}{1+\bar{m}^{2}}\right)\,.
\end{equation}
The complete flow equation for zero-momenta $4$-point coupling write explicitly as:
\bea\label{flowfour}
\dot{\Gamma}^{(4)}=-\sum_{\vec p}\dot r_s(\vec p\,) G^2_s(\vec p\,)\Big[\Gamma^{(6)}_{\vec p,\vec 0,\vec 0,\vec p,\vec 0,\vec 0}
\quad-2\sum_{\vec p\,'}\Gamma^{(4)}_{\vec p,\vec 0,\vec p\,',\vec 0}G_s(\vec p\,')\Gamma^{(4)}_{\vec p\,',\vec 0,\vec p,\vec 0}+2 G_s(\vec p\,)[\Gamma^{(4)}_{\vec p,\vec 0,\vec p,\vec 0}]^2\Big].\cr
\eea
Keeping only the melonic contributions, we get finally the following autonomous system  by using  the Litim's regulation:
\begin{align}
\left\{
    \begin{array}{ll}
       \beta_m&=-(2+\eta)\bar{m}^{2}-10\bar{\lambda}_{41}\,\frac{\pi^2}{(1+\bar{m}^{2})^2}\,\left(1+\frac{\eta}{6}\right)\,, \\
       \beta_{41}&=-2\eta \bar{\lambda}_{41}+4\bar{\lambda}_{41}^2 \,\frac{\pi^2}{(1+\bar{m}^{2})^3}\,\left(1+\frac{\eta}{6}\right)\Big[1-\pi^2\bar{\lambda}_{41}\left(\frac{1}{(1+\bar{m}^{2})^2}+\left(1+\frac{1}{1+\bar{m}^{2}}\right)\right)\Big]\,. \label{syst3}
    \end{array}
\right.
\end{align}
where the anomalous dimension is then given by:

\begin{equation}
\eta=4\bar{\lambda}_{41}\pi^2\frac{(1+\bar{m}^{2})^2-\frac{1}{2}\bar{\lambda}_{41}\pi^2(2+\bar{m}^{2})}{(1+\bar{m}^{2})^2\Omega(\bar{\lambda}_{41},\bar{m}^{2})+\frac{(2+\bar{m}^{2})}{3}\bar{\lambda}_{41}^2\pi^4}\,.
\end{equation}
\begin{figure}\begin{center}
\includegraphics[scale=0.4]{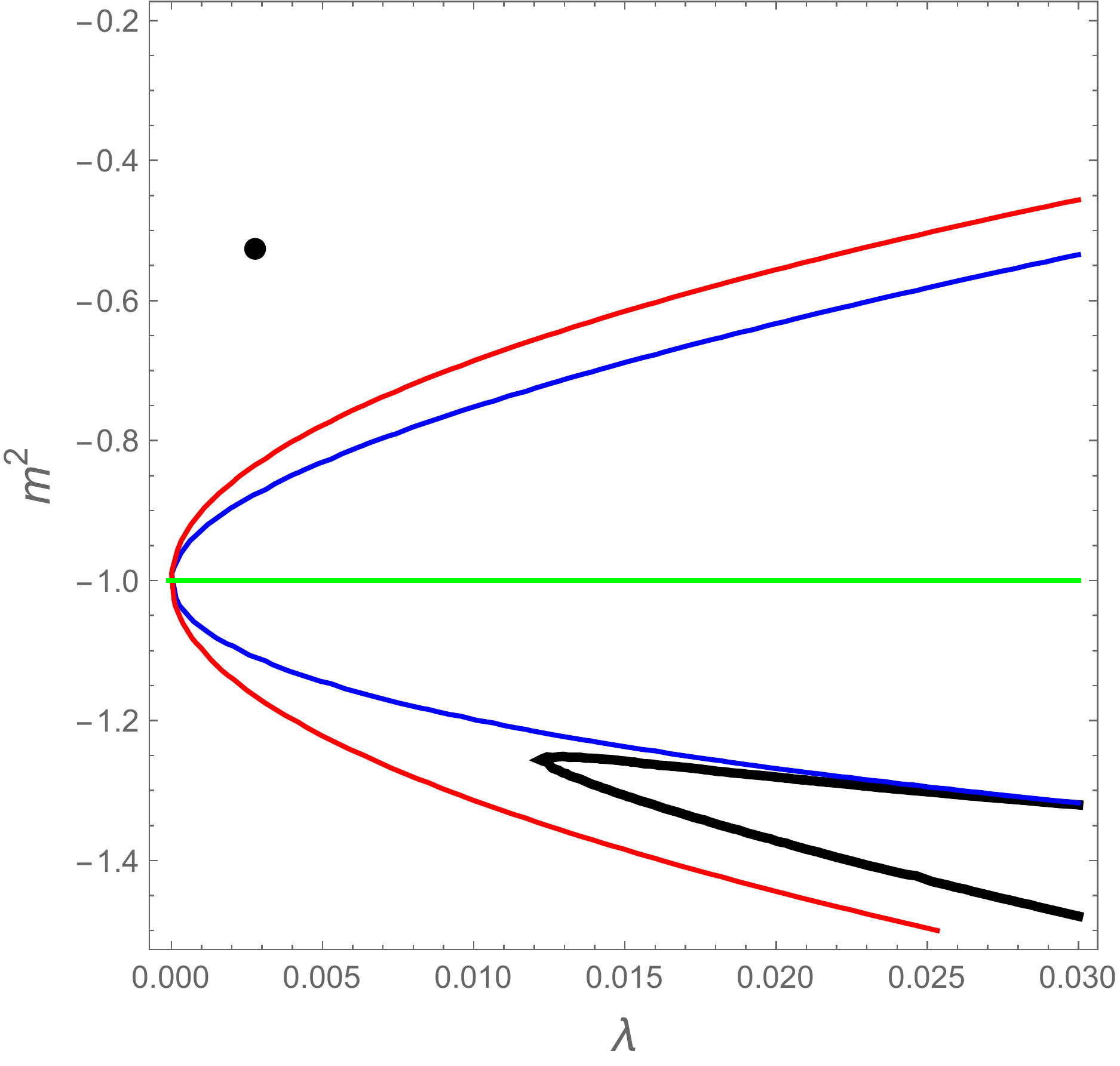} 
\caption{The relevant lines over the maximally extended region $I^\prime$, bounded at the bottom with the singularity line $m^{2}=-1$ (in green). The blue and red curves correspond respectively to the equations $L=0$ and $\Omega=0$. Moreover, the black point correspond to the numerical non-Gaussian fixed point, so far from the two previous physical curves. }\label{fig3}\end{center}
\end{figure}

The new anomalous dimension has two properties which distinguish him from its truncation version. First of all, as announced, the singularity line $\Omega=0$ moves toward the $\bar{\lambda}_{41}$ axis, extending the symmetric phase domain.
In fact, the improvement is \textit{maximal}, the critical line being deported under the singularity line $\bar{m}^{2}=-1$.  In standard interpretations \cite{Lahoche:2018oeo}, the presence of the region $II$ is generally assumed to come from a bad expansion of the effective average action around vanishing means field, becoming a spurious vacuum in this region.

However the EVE method show that the singularity line obtained using truncation is completely discarded taking into account the momentum dependence of the effective vertex. The second improvement come from the fact that the anomalous dimension may be negative, and vanish on the line of equation $L(\bar{\lambda}_{41},\bar{m}^{2})=0$, with:
\beq
L(\bar{\lambda}_{41},\bar{m}^{2}):=(1+\bar{m}^{2})^2-\frac{1}{2}\bar{\lambda}_{41}\pi^2(2+\bar{m}^{2})\,.
\eeq
 Interestingly, there are now two lines in the maximally extended region $I^\prime$ where physical fixed points are expected. However, numerical integrations, show that the improved flow equations admit a non-Gaussian fixed point $\tilde p_+$, which is  numerically very close from the fixed point $p_+$  obtained in the truncation method i.e. $\tilde p_+\approx p_+$, and then unphysical as well. 
The other solutions are:
\bea
p_{0}=(\bar{m}^2=-1.28,\bar{\lambda}_{41}=0.025),\quad p_{1}=(\bar{m}^2=1.96,\bar{\lambda}_{41}=1.10),\quad 
\eea
For $p_0$ we have $\bar{m}^2<-1$. This fixed point cannot be taking into account by considering all the explanation given in the section \eqref{sec2}. $p_1$ have the following critical exponent $\theta_1=-2.8-4.2i$, $\theta_2=-2.8+4.2i$. This fixed point is IR attractive and lives in the same region like $p_+$. Finally all the fixed point discovered from EVE method violate the Ward identities.

\subsection{Exploration of the physical phase space} 
In this section we  will show that the EVE method leads to an alternative first order phase transition scenario,  despite the fact that the fixed point $p_+$ is discarded. In the second time we also prove that this new behavior is only observed using EVE method  and can not be obtained  by implementing the  usual truncation as approximation.\\
{\bf 1})\,\,
Despite the fact that the constraint equation \eqref{contrainte1} is not compatible with the fixed point $p_+=(-0.52,0.0028)$, this is not the end of the history. The constraint $\mathcal{C}=0$ given by equation \eqref{contrainte1} define a one-dimensional subspace, say $\mathcal{E}_\mathcal{C}$ into the whole bi-dimensional phase space $(\bar{\lambda}_{41},\bar{m}^2)$. Obviously, the Ward identity will be violated everywhere except along this one-dimensional subspace $\mathcal{E}_\mathcal{C}$; for this reason we call \textit{physical phase space} this subspace. \\

\noindent 
Solving $\mathcal{C}=0$ with respect to $\bar{m}^2$, we can extract the coupling constant $\bar\lambda_{41}$ as function of the renormalized mass parameter $\bar m^2$. After a few handing computation, we get:
\bea\label{solnew3}
\bar\lambda_{41}^3=0,\quad \text{or}\quad \bar\lambda_{41}= \frac{(m+1)^2 (3 m (m+3)-10)}{\pi ^2 (m (m+7)+2)}:=f(\bar m^2).
\eea 
These solutions provides only one non-trivial parametrized equation for the physical subspace $\mathcal{E}_\mathcal{C}$ : $\bar\lambda_{41}=f(\bar m^2)$. Interestingly, it is not hard to cheek that the presence of the factor $(1+\bar{m}^2)^2$ in the numerator cancel all the formal divergences occurring for $\bar{m}^2=-1$, such that the flow becomes regular at this point. However, other divergences occurs, one of them being common to each beta functions. To understand the structure of the effective flow into the physical subspace, we have to insert the solutions \eqref{solnew3} into the flow equations \eqref{syst3}. However, even to do this, let us discuss the solution \eqref{solnew3} in a few words. Because the theory is asymptotically free, we may expect that $\bar{m}^2$ and $\bar{\lambda}_{41}$ have to vanish simultaneously. What we know is that, in the vicinity of the Gaussian fixed point $\bar{m}^2=\bar{\lambda}_{41}=0$, the constraint $\mathcal{C}=0$ is approximately satisfied. For instance, up to $\bar{\lambda}^3_{41}$ contributions, the equation \eqref{contrainte1} reduces as:
\begin{equation}
\mathcal{C}=\beta_{41}+\eta \bar{\lambda}_{41}=0
\end{equation}
which is identically satisfied from the one-loop beta equation $\beta_{41}=-\eta \bar{\lambda}_{41}$ -- see \eqref{syst3}. As a result, in a small domain around $(\bar{m}^2,\bar{\lambda}_{41})=(0,0)$, the flow behaves approximately according the Ward constraint, but as soon as the flow leaves this region, the Ward constraint is violated, except along the $\mathcal{E}_\mathcal{C}$, where it hold strictly. Note that, for $\bar{m}^2=0$, the value of $\bar{\lambda}_{41}$ is so large ($\bar{\lambda}_{41}\approx 1.9$), and far away from the vicinity of the Gaussian fixed point. \\

\noindent
Now, let us move on to the solutions \eqref{solnew3}. The solution $\bar{\lambda}_{41}=0$ corresponds to trivial flow, $\eta=0$ and:
\bea
\beta_{m}=-2\bar{m}^2,\quad 
\beta_{41}=0\,. 
\eea
On the other hand, inserting the non-trivial solution $\bar\lambda_{41}=f(\bar m^2)$, we get:
\bea
\eta(\bar m^2)=\frac{24 (\bar{m}^2 (\bar{m}^2+7)+2)}{\bar{m}^2 (3 \bar{m}^2 (\bar{m}^2 (\bar{m}^2+6)+1)-56)+68}-6.
\eea
and :
\bea
&&\beta_{m}=\frac{4 (\bar{m}^2 (\bar{m}^2 (3 \bar{m}^2 (\bar{m}^2 (\bar{m}^2+6)-1)-128)-34)+100)}{\bar{m}^2 (3 \bar{m}^2 (\bar{m}^2 (\bar{m}^2+6)+1)-56)+68},
\eea
\bea
\beta_{41}=\frac{4 (\bar{m}^2+1) (10-3 \bar{m}^2 (\bar{m}^2+3))^2 (\bar{m}^2 (\bar{m}^2 (3 \bar{m}^2 (\bar{m}^2 (\bar{m}^2+7)+3)-157)-104)+92)}{\pi ^2 (\bar{m}^2 (\bar{m}^2+7)+2)^2 (\bar{m}^2 (3 \bar{m}^2 (\bar{m}^2 (\bar{m}^2+6)+1)-56)+68)}.\cr
&&
\eea
\begin{center}
\begin{equation*}
\underset{a}{\vcenter{\hbox{\includegraphics[scale=0.35]{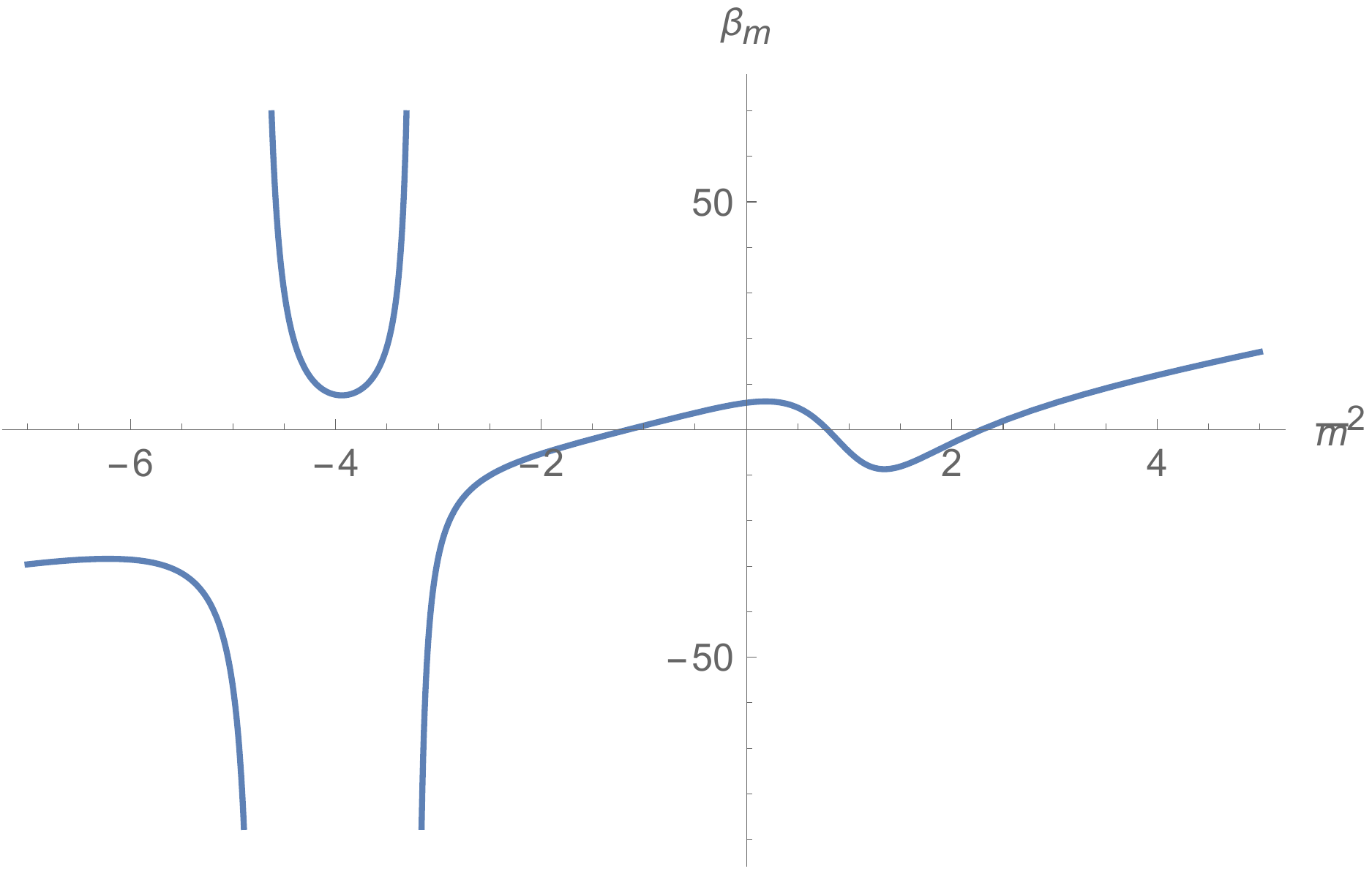}}}} \quad \underset{b}{\vcenter{\hbox{\includegraphics[scale=0.35]{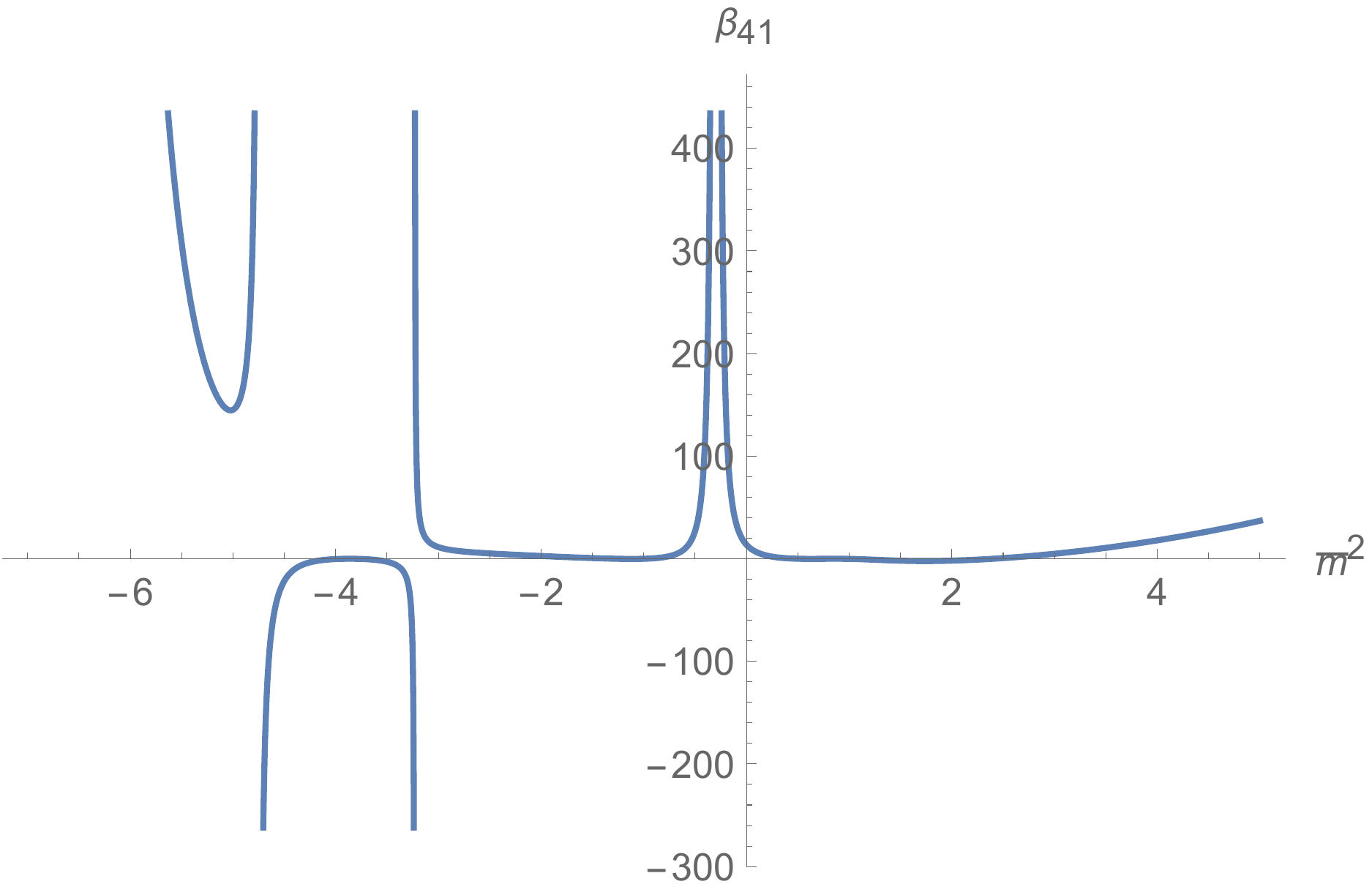}}}} 
\end{equation*}
\captionof{figure}{ (a)\,\,Plot of $\beta_m$ as function of $\bar m^2$ using the constraint equation. We get the singularity at the points $\bar m^2_{\text{div}1}=-4.75$, $\bar m^2_{\text{div}2}=-3.23$ corresponding to the coupling value $\bar \lambda_{\text{div}1}=-2.45$ and $\bar \lambda_{\text{div}2}=-0.38$.  (b)\,\, Plot of $\beta_{41}$ as function of $\bar m^2$ in the physical phase space. The  same singularity  points are identified. Note that the singularity point $\bar m^2_{0}=-0.29$ which appears in the denominator of $\beta_{41}$ do not implies a singularity for $\beta_m$. This point is reminescent of the  first order phase transition in the domain $\bar{m}^2\in]-\infty,-0.29]$. }\label{phase}
\end{center}
As announced, the divergences at the value $\bar{m}^2=-1$ has been discarded. However, some new divergences occurs. First of all, the equation for $\mathcal{E}_\mathcal{C}$ becomes singular for the value $\bar{m}_0^2=-0.29$. This singularity comes from  the denominator of $f(\bar{m}^2)$. Note that $f(\bar{m}^2)$ is such that for the small $\epsilon>0$, $f(\bar{m}_0^2-\epsilon)>0$ and $\epsilon>0$, $f(\bar{m}_0^2+\epsilon)<0$. This singularity is reminescent to the first order phase transition. A second singularity occurs for the values $\bar m^2_{\text{div}1}=-4.75$, $\bar m^2_{\text{div}2}=-3.23$, which is common for $\eta$, $\beta_m$ and $\beta_{41}$. We now discuss this picture. To this end, let us examine the points at which the beta function vanish. We get:
\bea
\beta_m(\bar m_1^2)=0\Rightarrow \bar m_1^2= 2.29,\quad \bar{m}_1^2= -1.14,\quad \bar{m}_1^2 = 0.78
\eea
\bea
\beta_{41}(\bar m_2^2)=0\Rightarrow &&\bar m_2^2=-3.86,\quad \bar{m}_2^2=2.40,\quad \bar{m}_2^2=-1.25,\cr
&&\bar{m}_2^2=0.86,  \quad\bar{m}_2^2=0.51,\quad \bar m_2^2=-1.
\eea
Because $\bar m_1^2\neq \bar m_2^2$, we recover our previous conclusion, in the whole theory space $(\bar{\lambda},\bar{m}^2)$, no fixed point can be found using the exact FRG with the EVE method taking into account the Ward constraint. Finally at the point $\bar m_0^2\approx -0.29$ on the projected phase space $\mathcal{E}_{\mathcal{C}}$, the discontinuity of $f(\bar{m}^2)$ implies the discontinuity of the effective action $\Gamma_s$. The flow into the physical phase space change the direction at this point, pointing toward positive mass direction for $\bar{m}^2>\bar m_0^2$ and toward the negative mass direction for $\bar{m}^2<\bar m_0^2$. In the last case, the flow continues on this way and reaches the singularity, where the flow becomes undefined. Both, these two features are reminiscent of a first order phase transition
on the physical phase space -- the singularity may indicate
a point at which the effective action becomes undefined, or where the expansion around the null vacuum fails to exist – the last statement having to
be rigorously investigated.

\noindent
The same analysis may be performed when we consider the following prescription:  by extracting the mass parameter $\bar m^2$ as function of the constant $\bar\lambda_{41}$: $(\bar m^2=g(\bar\lambda_{41}))$ in the constraint equation and solve the $\beta$-function of the coupling. In this case, the coupling becomes the parameter, and for the points $\lambda_{\text{div}1}=-2.45$ and $\bar \lambda_{\text{div}2}=-0.38$ we get a singularity corresponding to the values $\bar m^2_{\text{div}1}=-4.75$, $\bar m^2_{\text{div}2}=-3.23$ (see Figure  \eqref{phase}b).  Note that around $\bar{m}_0^2$, the coupling becomes very small :
\begin{equation}
f(\bar{m}^2_0)\approx 0.0077\,,
\end{equation}
and we reach a new perturbative regime for small $\bar{\lambda}_{41}$ and small $(1+\bar{m}^2)$. \\
{\bf 2})\, When we investigated the truncation method, we do not performed such a discussion. To compare the methods, let us consider the same strategy for the phase space described with the truncation method. Solving the constraint $\mathcal{C}=0$, we get:
\bea
\bar \lambda_{41}^3=0 \mbox{ or }  \bar \lambda_{41}=\frac{11 (1 + \bar m^2)^2}{5 \pi^2}.
\eea
By replacing this solution $\bar \lambda_{41}^3=0 $ in the flow equations of mass and coupling \eqref{flownew} we get
\bea
\beta_m=-2\bar m^2,\quad \beta_{41}=0.
\eea
Now setting $\beta_m=0=\beta_{41}$, only the Gaussian fixed point $(\bar m^*=0,\bar \lambda^*_{41}=0)$ survives. Also the last solution leads to
\bea
\beta_m=\frac{4}{9} (12 \bar{m}^2+11),\quad  \beta_{41}=\frac{484 (\bar{m}^2+1) (15 \bar{m}^2+13)}{225 \pi ^2}.
\eea
One more time, we recover that no solutions such that $\beta_m=0=\beta_{41}$ exist. Moreover, we recover that $\beta_m$ vanish for a negative mass value, not so far from $\bar{m}^2=-1$; and that the singularity at this value has been completely discarded from the solution of the Ward constraint. However, the common singularity of the beta functions as some other aspects of the previous flow equations are not reproduced in the truncation framework. The nature of the singularities, for $\bar m^2_{\text{div}1}=-4.75$, $\bar m^2_{\text{div}2}=-3.23$ remains mysterious in our formalism. Obviously they are a consequence of the improvement coming from the EVE method, and their understanding may be increasing our knowledge about the behavior of the TGFT renormalization group flow.

\section{Conclusion}\label{sec5}
In this manuscript we have studied with different methods the FRG applied to TGFT. 
First we have derived the Wilson-Polchinski  equation and given the perturbative solution. In the second time we derived the Wetterich flow equation using the usual approximation called truncation. The analytic solution of this equation is given. We get a fixed point denoted by $p_+$. Then we investigated the Ward identities as a new constraint along the flow and showed that the fixed point $p_+$ violates this constraint. Finally we improve the study of FRG by replacing the truncation method by the so called EVE. The flow equation is improved and the corresponding solution $\tilde p_+$ is not so far from $p_+$ i.e. $\tilde p_+\approx p_+$ . However, the Ward identities are strongly violated
at this fixed point and therefore this unique fixed point seems to be unphysical.  We have also showed the importance of EVE method in the sense that, despite the fact that the fixed point $p_+$ needs to be discarded, a first order phase transition exists so far from this point in the subspace $\mathcal{E}_C$ of the theory space. We have showed that this new behavior  can not be observed using the truncation as approximation. \\

In this review we focus on the EVE method for the melonic approximation, and especially on the quartic melonic just-renormalizable sector. The complete quartic sector, including all the connected quartic bubbles has already been considered in a complementary work \cite{Lahoche:2018oeo}, and the conclusion about the incompatibility with nonperturbative fixed points and Ward identities hold. The graphs added to the quartic melonic ones to complete the quartic sector have been called \textit{pseudo-melons} due to the similarities of their respective leading order Feynman graphs. Finally, even if we expect that some aspects of the EVE method improve the standard truncation method, some limitations have to be addressed for future works. In particular our investigations
are limited on the symmetric phase, ensuring convergence of any expansion around
vanishing classical means field. Moreover, we have retained only the first terms in the derivative
expansion of the 2-point function and only considered the local potential approximation,
i.e. potentials which can be expanded as an infinite sum of connected melonic (and pseudo-melonic) interactions. Finally, a rigorous investigation of the behavior of the renormalization group flow into the physical phase space has to be addressed in the continuation of current works on this topic.


\begin{thebibliography}{16}






\bibitem{Rovelli:1997yv} 
  C.~Rovelli,
  ``Loop quantum gravity,''
  Living Rev.\ Rel.\  {\bf 1}, 1 (1998)
  doi:10.12942/lrr-1998-1
  [gr-qc/9710008].


\bibitem{Rovelli:1998gg} 
  C.~Rovelli and P.~Upadhya,
  ``Loop quantum gravity and quanta of space: A Primer,''
  gr-qc/9806079.

\bibitem{Ambjorn:1992eh} 
  J.~Ambjorn, Z.~Burda, J.~Jurkiewicz and C.~F.~Kristjansen,
  ``Quantum gravity represented as dynamical triangulations,''
  Acta Phys.\ Polon.\ B {\bf 23}, 991 (1992).

\bibitem{Ambjorn:1995jt} 
  J.~Ambjorn,
  ``Quantum gravity represented as dynamical triangulations,''
  Class.\ Quant.\ Grav.\  {\bf 12}, 2079 (1995).
  doi:10.1088/0264-9381/12/9/002


\bibitem{Ambjorn:2013tki} 
  J.~Ambjørn, A.~Görlich, J.~Jurkiewicz and R.~Loll,
  ``Quantum Gravity via Causal Dynamical Triangulations,''
  doi:10.1007/978-3-642-41992-8-34
  arXiv:1302.2173 [hep-th].




\bibitem{Connes:1990qp} 
  A.~Connes and J.~Lott,
``Particle Models and Noncommutative Geometry (Expanded
Version),''
  Nucl.\ Phys.\ Proc.\ Suppl.\  {\bf 18B}, 29 (1991).
  doi:10.1016/0920-5632(91)90120-4


\bibitem{Aastrup:2006ib} 
  J.~Aastrup and J.~M.~Grimstrup,
``Intersecting connes noncommutative geometry with quantum
gravity,''
  Int.\ J.\ Mod.\ Phys.\ A {\bf 22}, 1589 (2007)
  doi:10.1142/S0217751X07035306
  [hep-th/0601127].






\bibitem{Oriti:2006ar} 
  D.~Oriti,
``A Quantum field theory of simplicial geometry and the
emergence of spacetime,''
  J.\ Phys.\ Conf.\ Ser.\  {\bf 67}, 012052 (2007)
  doi:10.1088/1742-6596/67/1/012052
  [hep-th/0612301].


\bibitem{deCesare:2016rsf} 
  M.~de Cesare, A.~G.~A.~Pithis and M.~Sakellariadou,
``Cosmological implications of interacting Group Field Theory
models: cyclic Universe and accelerated expansion,''
  Phys.\ Rev.\ D {\bf 94}, no. 6, 064051 (2016)
  doi:10.1103/PhysRevD.94.064051
  [arXiv:1606.00352 [gr-qc]].



\bibitem{Gielen:2016dss} 
  S.~Gielen and L.~Sindoni,
``Quantum Cosmology from Group Field Theory Condensates: a
Review,''
  SIGMA {\bf 12}, 082 (2016)
  doi:10.3842/SIGMA.2016.082
  [arXiv:1602.08104 [gr-qc]].



\bibitem{Gielen:2017eco} 
  S.~Gielen and D.~Oriti,
  ``Cosmological perturbations from full quantum gravity,''
  arXiv:1709.01095 [gr-qc].

%
%
%




\bibitem{Oriti:2014yla} 
  D.~Oriti, J.~P.~Ryan and J.~Thurigen,
  ``Group field theories for all loop quantum gravity,''
  New J.\ Phys.\  {\bf 17}, no. 2, 023042 (2015)
  doi:10.1088/1367-2630/17/2/023042
  [arXiv:1409.3150 [gr-qc]].




\bibitem{Gurau:2009tw} 
  R.~Gurau,
  ``Colored Group Field Theory,''
  Commun.\ Math.\ Phys.\  {\bf 304}, 69 (2011)
  doi:10.1007/s00220-011-1226-9
  [arXiv:0907.2582 [hep-th]].
 
\bibitem{Rivasseau:2016rgt} 
  V.~Rivasseau,
  ``Constructive Tensor Field Theory,''
  arXiv:1603.07312 [math-ph].


\bibitem{Rivasseau:2016zco} 
  V.~Rivasseau,
  ``Random Tensors and Quantum Gravity,''
  arXiv:1603.07278 [math-ph].



\bibitem{Rivasseau:2014ima} 
  V.~Rivasseau,
  ``The Tensor Theory Space,''
  Fortsch.\ Phys.\  {\bf 62}, 835 (2014)
  doi:10.1002/prop.201400057
  [arXiv:1407.0284 [hep-th]].

\bibitem{Rivasseau:2013uca} 
  V.~Rivasseau,
  ``The Tensor Track, III,''
  Fortsch.\ Phys.\  {\bf 62}, 81 (2014)
  doi:10.1002/prop.201300032
  [arXiv:1311.1461 [hep-th]].


\bibitem{Rivasseau:2016wvy} 
  V.~Rivasseau,
  ``The Tensor Track, IV,''
  arXiv:1604.07860 [hep-th].


\bibitem{Rivasseau:2012yp} 
  V.~Rivasseau,
  ``The Tensor Track: an Update,''
  arXiv:1209.5284 [hep-th].


\bibitem{Gurau:2011xq} 
  R.~Gurau,
``The complete 1/N expansion of colored tensor models in
arbitrary dimension,''
  Annales Henri Poincare {\bf 13}, 399 (2012)
  doi:10.1007/s00023-011-0118-z
  [arXiv:1102.5759 [gr-qc]].



\bibitem{Gurau:2010ba} 
  R.~Gurau,
  ``The 1/N expansion of colored tensor models,''
  Annales Henri Poincare {\bf 12}, 829 (2011)
  doi:10.1007/s00023-011-0101-8
  [arXiv:1011.2726 [gr-qc]].


\bibitem{Gurau:2013pca} 
  R.~Gurau,
  ``The 1/N Expansion of Tensor Models Beyond Perturbation Theory,''
  Commun.\ Math.\ Phys.\  {\bf 330}, 973 (2014)
  doi:10.1007/s00220-014-1907-2
  [arXiv:1304.2666 [math-ph]].









\bibitem{Carrozza:2012uv} 
  S.~Carrozza, D.~Oriti and V.~Rivasseau,
``Renormalization of Tensorial Group Field Theories: Abelian
U(1) Models in Four Dimensions,''
  Commun.\ Math.\ Phys.\  {\bf 327}, 603 (2014)
  doi:10.1007/s00220-014-1954-8
  [arXiv:1207.6734 [hep-th]].


\bibitem{Carrozza:2013mna} 
  S.~Carrozza,
``Tensorial methods and renormalization in Group Field
Theories,''
  doi:10.1007/978-3-319-05867-2
  arXiv:1310.3736 [hep-th].



\bibitem{Carrozza:2013wda} 
  S.~Carrozza, D.~Oriti and V.~Rivasseau,
``Renormalization of a SU(2) Tensorial Group Field Theory in
Three Dimensions,''
  Commun.\ Math.\ Phys.\  {\bf 330}, 581 (2014)
  doi:10.1007/s00220-014-1928-x
  [arXiv:1303.6772 [hep-th]].


\bibitem{Geloun:2013saa} 
  J.~Ben Geloun,
``Renormalizable Models in Rank $d\geq 2$ Tensorial Group Field
Theory,''
  Commun.\ Math.\ Phys.\  {\bf 332}, 117 (2014)
  doi:10.1007/s00220-014-2142-6
  [arXiv:1306.1201 [hep-th]].
\bibitem{Lahoche:2015tqa} 
  V.~Lahoche and D.~Oriti,
``Renormalization of a tensorial field theory on the homogeneous
space SU(2)/U(1),''
  arXiv:1506.08393 [hep-th].


\bibitem{Lahoche:2015ola} 
  V.~Lahoche, D.~Oriti and V.~Rivasseau,
``Renormalization of an Abelian Tensor Group Field Theory:
Solution at Leading Order,''
  JHEP {\bf 1504}, 095 (2015)
  doi:10.1007/JHEP04(2015)095
  [arXiv:1501.02086 [hep-th]].


\bibitem{Geloun:2012bz} 
  J.~Ben Geloun and E.~R.~Livine,
  ``Some classes of renormalizable tensor models,''
  J.\ Math.\ Phys.\  {\bf 54}, 082303 (2013)
  doi:10.1063/1.4818797
  [arXiv:1207.0416 [hep-th]].


\bibitem{Samary:2012bw} 
  D.~Ousmane~Samary and F.~Vignes-Tourneret,
``Just Renormalizable TGFT's on $U(1)^{d}$ with Gauge
Invariance,''
  Commun.\ Math.\ Phys.\  {\bf 329}, 545 (2014)
  doi:10.1007/s00220-014-1930-3
  [arXiv:1211.2618 [hep-th]].


\bibitem{BenGeloun:2012pu} 
  J.~Ben Geloun and D.~Ousmane.~Samary,
``3D Tensor Field Theory: Renormalization and One-loop
$\beta$-functions,''
  Annales Henri Poincare {\bf 14}, 1599 (2013)
  doi:10.1007/s00023-012-0225-5
  [arXiv:1201.0176 [hep-th]].

\bibitem{BenGeloun:2011rc} 
  J.~Ben Geloun and V.~Rivasseau,
  ``A Renormalizable 4-Dimensional Tensor Field Theory,''
  Commun.\ Math.\ Phys.\  {\bf 318}, 69 (2013)
  doi:10.1007/s00220-012-1549-1
  [arXiv:1111.4997 [hep-th]].



\bibitem{Lahoche:2015ola} 
  V.~Lahoche, D.~Oriti and V.~Rivasseau,
``Renormalization of an Abelian Tensor Group Field Theory:
Solution at Leading Order,''
  JHEP {\bf 1504}, 095 (2015)
  doi:10.1007/JHEP04(2015)095
  [arXiv:1501.02086 [hep-th]].







\bibitem{BenGeloun:2017xbd} 
  J.~Ben Geloun and R.~Toriumi,
``Renormalizable Enhanced Tensor Field Theory: The quartic
melonic case,''
  arXiv:1709.05141 [hep-th].


\bibitem{Geloun:2011cy} 
  J.~Ben Geloun and V.~Bonzom,
``Radiative corrections in the Boulatov-Ooguri tensor model: The
2-point function,''
  Int.\ J.\ Theor.\ Phys.\  {\bf 50}, 2819 (2011)
  doi:10.1007/s10773-011-0782-2
  [arXiv:1101.4294 [hep-th]].







\bibitem{BenGeloun:2012yk} 
  J.~Ben Geloun,
``Two and four-loop $\beta$-functions of rank 4 renormalizable
tensor field theories,''
  Class.\ Quant.\ Grav.\  {\bf 29}, 235011 (2012)
  doi:10.1088/0264-9381/29/23/235011
  [arXiv:1205.5513 [hep-th]].



\bibitem{Samary:2013xla} 
  D.~Ousmane Samary,
``Beta functions of   $U(1)^d$ gauge invariant just renormalizable
tensor models,''
  Phys.\ Rev.\ D {\bf 88}, no. 10, 105003 (2013)
  doi:10.1103/PhysRevD.88.105003
  [arXiv:1303.7256 [hep-th]].

\bibitem{Rivasseau:2015ova} 
  V.~Rivasseau,
  ``Why are tensor field theories asymptotically free?,''
  EPL {\bf 111}, no. 6, 60011 (2015)
  doi:10.1209/0295-5075/111/60011
  [arXiv:1507.04190 [hep-th]].



\bibitem{Carrozza:2014rba} 
  S.~Carrozza,
``Discrete Renormalization Group for SU(2) Tensorial Group Field
Theory,''
Ann. Inst. Henri Poincar\'e Comb. Phys. Interact. 2 (2015),
49-112
  doi:10.4171/AIHPD/15
  [arXiv:1407.4615 [hep-th]].

\bibitem{Wilson}
K. Wilson,  ``Renormalization Group and Critical Phenomena. I'',Phys. Rev. B4 (1971) 3174; K. Wilson and J. Kogut, ``Renormalization Group and Critical Phenomena. II.'',  Phys. Rep. 12C
(1975) 75.

\bibitem{Polchinshi}
Polchinski, J. (1984), "Renormalization and Effective Lagrangians", Nucl. Phys. B, 231 (2): 269.

\bibitem{Oriti:2018dsg} 
  D.~Oriti,
  ``Levels of spacetime emergence in quantum gravity,''
  arXiv:1807.04875 [physics.hist-ph].

\bibitem{Oriti:2013jga} 
  D.~Oriti,
``Disappearance and emergence of space and time in quantum
gravity,''
  Stud.\ Hist.\ Phil.\ Sci.\ B {\bf 46}, 186 (2014)
  doi:10.1016/j.shpsb.2013.10.006
  [arXiv:1302.2849 [physics.hist-ph]].

\bibitem{Markopoulou:2007jf} 
  F.~Markopoulou,
  ``Conserved quantities in background independent theories,''
  J.\ Phys.\ Conf.\ Ser.\  {\bf 67}, 012019 (2007)
  doi:10.1088/1742-6596/67/1/012019
  [gr-qc/0703027].


\bibitem{Wilkinson:2015fja} 
  S.~A.~Wilkinson and A.~D.~Greentree,
``Geometrogenesis under Quantum Graphity: problems with the
ripening Universe,''
  Phys.\ Rev.\ D {\bf 92}, no. 8, 084007 (2015)
  doi:10.1103/PhysRevD.92.084007
  [arXiv:1506.07588 [gr-qc]].





\bibitem{Geloun:2016qyb} 
  J.~B.~Geloun, R.~Martini and D.~Oriti,
``Functional Renormalisation Group analysis of Tensorial Group
Field Theories on $\mathbb{R}^d$,''
  arXiv:1601.08211 [hep-th].

\bibitem{Geloun:2015qfa} 
  J.~B.~Geloun, R.~Martini and D.~Oriti,
``Functional Renormalization Group analysis of a Tensorial Group
Field Theory on $\mathbb{R}^3$,''
  Europhys.\ Lett.\  {\bf 112}, no. 3, 31001 (2015)
  doi:10.1209/0295-5075/112/31001
  [arXiv:1508.01855 [hep-th]].




\bibitem{Benedetti:2015yaa} 
  D.~Benedetti and V.~Lahoche,
``Functional Renormalization Group Approach for Tensorial Group
Field Theory: A Rank-6 Model with Closure Constraint,''
  arXiv:1508.06384 [hep-th].

\bibitem{Benedetti:2014qsa} 
  D.~Benedetti, J.~Ben Geloun and D.~Oriti,
``Functional Renormalisation Group Approach for Tensorial Group
Field Theory: a Rank-3 Model,''
  JHEP {\bf 1503}, 084 (2015)
  doi:10.1007/JHEP03(2015)084
  [arXiv:1411.3180 [hep-th]].


\bibitem{BenGeloun:2018ekd} 
  J.~Ben Geloun, T.~A.~Koslowski, D.~Oriti and A.~D.~Pereira,
``Functional Renormalization Group analysis of rank 3 tensorial
group field theory: The full quartic invariant truncation,''
  Phys.\ Rev.\ D {\bf 97}, no. 12, 126018 (2018)
  doi:10.1103/PhysRevD.97.126018
  [arXiv:1805.01619 [hep-th]].





\bibitem{Carrozza:2016tih} 
  S.~Carrozza and V.~Lahoche,
``Asymptotic safety in three-dimensional SU(2) Group Field
Theory: evidence in the local potential approximation,''
  Class.\ Quant.\ Grav.\  {\bf 34}, no. 11, 115004 (2017)
  doi:10.1088/1361-6382/aa6d90
  [arXiv:1612.02452 [hep-th]].


\bibitem{Lahoche:2016xiq} 
  V.~Lahoche and D.~Ousmane Samary,
``Functional renormalization group for the U(1)-T$_5^6$
tensorial group field theory with closure constraint,''
  Phys.\ Rev.\ D {\bf 95}, no. 4, 045013 (2017)
  doi:10.1103/PhysRevD.95.045013
  [arXiv:1608.00379 [hep-th]].






\bibitem{Carrozza:2017vkz} 
  S.~Carrozza, V.~Lahoche and D.~Oriti,
``Renormalizable Group Field Theory beyond melonic diagrams: an
example in rank four,''
  Phys.\ Rev.\ D {\bf 96}, no. 6, 066007 (2017)
  doi:10.1103/PhysRevD.96.066007
  [arXiv:1703.06729 [gr-qc]].




\bibitem{BenGeloun:2011xu} 
  J.~Ben Geloun,
 ``Ward-Takahashi identities for the colored Boulatov model,''
{\bf 44}, 415402 (2011)
  doi:10.1088/1751-8113/44/41/415402
  [arXiv:1106.1847 [hep-th]].


\bibitem{Perez-Sanchez:2016zbh} 
  C.~I.~Pérez-Sánchez,
 ``The full Ward-Takahashi Identity for colored tensor
models,''
  doi:10.1007/s00220-018-3103-2
  arXiv:1608.08134 [math-ph].
  
\bibitem{Lahoche:2018ggd} 
  V.~Lahoche and D.~Ousmane.~Samary,
  ``Ward identity violation for melonic $T^4$-truncation,''
  arXiv:1809.06081 [hep-th].



\bibitem{Lahoche:2018oeo} 
  V.~Lahoche and D.~Ousmane.~Samary,
  ``Nonperturbative renormalization group beyond melonic sector: The Effective Vertex Expansion method for group fields theories,''
  arXiv:1809.00247 [hep-th].


\bibitem{Lahoche:2018vun} 
  V.~Lahoche and D.~Ousmane.~Samary,
  ``Unitary symmetry constraints on tensorial group field theory renormalization group flow,''
  Class.\ Quant.\ Grav.\  {\bf 35}, no. 19, 195006 (2018)
  doi:10.1088/1361-6382/aad83f
  [arXiv:1803.09902 [hep-th]].



\bibitem{Samary:2014tja} 
  D.~Ousmane.~Samary,
  ``Closed equations of the two-point functions for tensorial group field theory,''
  Class.\ Quant.\ Grav.\  {\bf 31}, 185005 (2014)
  doi:10.1088/0264-9381/31/18/185005
  [arXiv:1401.2096 [hep-th]].


\bibitem{Samary:2014oya} 
D.~Ousmane Samary, C.~I.~P\'erez-S\'anchez, F.~Vignes-Tourneret
and R.~Wulkenhaar,
``Correlation functions of a just renormalizable tensorial group
field theory: the melonic approximation,''
  Class.\ Quant.\ Grav.\  {\bf 32}, no. 17, 175012 (2015)
  doi:10.1088/0264-9381/32/17/175012
  [arXiv:1411.7213 [hep-th]].




%




\bibitem{Bonzom:2012wa} 
  V.~Bonzom,
  ``New 1/N expansions in random tensor models,''
  JHEP {\bf 1306}, 062 (2013)
  doi:10.1007/JHEP06(2013)062
  [arXiv:1211.1657 [hep-th]].




\bibitem{Bonzom:2015axa} 
  V.~Bonzom, T.~Delepouve and V.~Rivasseau,
``Enhancing non-melonic triangulations: A tensor model mixing
melonic and planar maps,''
  Nucl.\ Phys.\ B {\bf 895}, 161 (2015)
  doi:10.1016/j.nuclphysb.2015.04.004
  [arXiv:1502.01365 [math-ph]].


\bibitem{Bonzom:2012hw} 
  V.~Bonzom, R.~Gurau and V.~Rivasseau,
``Random tensor models in the large N limit: Uncoloring the
colored tensor models,''
  Phys.\ Rev.\ D {\bf 85}, 084037 (2012)
  doi:10.1103/PhysRevD.85.084037
  [arXiv:1202.3637 [hep-th]].





\bibitem{Gielen:2014uga} 
  S.~Gielen and D.~Oriti,
``Quantum cosmology from quantum gravity condensates:
cosmological variables and lattice-refined dynamics,''
  New J.\ Phys.\  {\bf 16}, no. 12, 123004 (2014)
  doi:10.1088/1367-2630/16/12/123004
  [arXiv:1407.8167 [gr-qc]].




\bibitem{Wilson:1971bg} 
  K.~G.~Wilson,
``Renormalization group and critical phenomena. 1.
Renormalization group and the Kadanoff scaling picture,''
  Phys.\ Rev.\ B {\bf 4}, 3174 (1971).
  doi:10.1103/PhysRevB.4.3174

%
\bibitem{Wilson:1971dh} 
  K.~G.~Wilson,
``Renormalization group and critical phenomena. 2. Phase space
cell analysis of critical behavior,''
  Phys.\ Rev.\ B {\bf 4}, 3184 (1971).
  doi:10.1103/PhysRevB.4.3184

\bibitem{Wetterich:1992yh} 
  C.~Wetterich,
  ``Exact evolution equation for the effective potential,''
  Phys.\ Lett.\ B {\bf 301}, 90 (1993)
  doi:10.1016/0370-2693(93)90726-X
  [arXiv:1710.05815 [hep-th]].

\bibitem{Wetterich:1989xg} 
  C.~Wetterich,
  ``Average Action and the Renormalization Group Equations,''
  Nucl.\ Phys.\ B {\bf 352}, 529 (1991).
  doi:10.1016/0550-3213(91)90099-J.


\bibitem{Wetterich:2001kra} 
  C.~Wetterich,
``Effective average action in statistical physics and quantum
field theory,''
  Int.\ J.\ Mod.\ Phys.\ A {\bf 16}, 1951 (2001)
  doi:10.1142/S0217751X01004591
  [hep-ph/0101178].

\bibitem{Schnoerr:2013bk} 
  D.~Schnoerr, I.~Boettcher, J.~M.~Pawlowski and C.~Wetterich,
``Error estimates and specification parameters for functional
renormalization,''
  Annals Phys.\  {\bf 334}, 83 (2013)
  doi:10.1016/j.aop.2013.03.013
  [arXiv:1301.4169 [cond-mat.quant-gas]].



\bibitem{Berges:2000ew} 
  J.~Berges, N.~Tetradis and C.~Wetterich,
``Nonperturbative renormalization flow in quantum field theory
and statistical physics,''
  Phys.\ Rept.\  {\bf 363}, 223 (2002)
  doi:10.1016/S0370-1573(01)00098-9
  [hep-ph/0005122].


\bibitem{Delamotte:2007pf} 
  B.~Delamotte,
  ``An Introduction to the nonperturbative renormalization group,''
  Lect.\ Notes Phys.\  {\bf 852}, 49 (2012)
  doi:10.1007/978-3-642-27320-9-2
  [cond-mat/0702365 [cond-mat.stat-mech]].








\bibitem{Tetradis:1995br} 
  N.~Tetradis and D.~F.~Litim,
``Analytical solutions of exact renormalization group
equations,''
  Nucl.\ Phys.\ B {\bf 464}, 492 (1996)
  doi:10.1016/0550-3213(95)00642-7
  [hep-th/9512073].




\bibitem{Litim:2000ci} 
  D.~F.~Litim,
  ``Optimization of the exact renormalization group,''
  Phys.\ Lett.\ B {\bf 486}, 92 (2000)
 doi:10.1016/S0370-2693(00)00748-6
  [hep-th/0005245].
  


\bibitem{Litim:2001dt} 
  D.~F.~Litim,
  ``Derivative expansion and renormalization group flows,''
  JHEP {\bf 0111}, 059 (2001)
  doi:10.1088/1126-6708/2001/11/059
  [hep-th/0111159].


\bibitem{Nagy:2012ef} 
  S.~Nagy,
  ``Lectures on renormalization and asymptotic safety,''
  Annals Phys.\  {\bf 350}, 310 (2014)
  doi:10.1016/j.aop.2014.07.027
  [arXiv:1211.4151 [hep-th]].

\bibitem{Blaizot:2005wd} 
  J.~P.~Blaizot, R.~Mendez-Galain and N.~Wschebor,
  ``Non perturbative renormalisation group and momentum dependence of n-point functions (I),''
  Phys.\ Rev.\ E {\bf 74}, 051116 (2006)
  doi:10.1103/PhysRevE.74.051116
  [hep-th/0512317].



\bibitem{Blaizot:2006vr} 
  J.~P.~Blaizot, R.~Mendez-Galain and N.~Wschebor,
  ``Non perturbative renormalization group and momentum dependence of n-point functions. II.,''
  Phys.\ Rev.\ E {\bf 74}, 051117 (2006)
  doi:10.1103/PhysRevE.74.051117
  [hep-th/0603163].


\bibitem{Defenu:2014jfa} 
N.~Defenu, P.~Mati, I.~G.~Marian, I.~Nandori and A.~Trombettoni,``Truncation
Effects in the Functional Renormalization Group
Study of Spontaneous Symmetry Breaking,''
  JHEP {\bf 1505}, 141 (2015)
  doi:10.1007/JHEP05(2015)141
  [arXiv:1410.7024 [hep-th]].




\end{thebibliography}
\end{document}